\documentclass[aps,pra,twocolumn,nofootinbib]{revtex4-2}

\usepackage[T1]{fontenc}
\usepackage{lmodern}
\usepackage{amsmath,amssymb,mathtools}
\usepackage{amsthm}
\usepackage{microtype}
\usepackage{xcolor}
\usepackage{tikz}
\usepackage{xurl}
\usepackage[colorlinks=true,allcolors=blue,hypertexnames=false]{hyperref}

\usetikzlibrary{
  arrows.meta,
  calc,
  fit,
  shapes.geometric
}

\definecolor{vthreeexact}{HTML}{006D9C}
\definecolor{vthreerelaxed}{HTML}{D89000}
\definecolor{vthreepath}{HTML}{7546A6}
\definecolor{vthreegreen}{HTML}{16856B}

\tikzset{
  v3 every figure/.style={
    font=\footnotesize,
    >=Latex,
    line cap=round,
    line join=round
  },
  v3 panel/.style={
    font=\bfseries\small,
    anchor=north west,
    inner sep=0pt
  },
  v3 annotation/.style={
    align=center,
    text=black!80,
    font=\scriptsize
  }
}
 
\newcommand{\Tr}{\operatorname{Tr}}
\newcommand{\1}{\openone}
\newcommand{\Istar}{I_{3322}^{\,*}}
\newcommand{\betaPV}{\beta_{\mathrm{PV}}}

\theoremstyle{plain}
\newtheorem{theorem}{Theorem}
\newtheorem{lemma}{Lemma}
\newtheorem{proposition}{Proposition}
\newtheorem{corollary}{Corollary}
\theoremstyle{definition}

\theoremstyle{remark}
\newtheorem{remark}{Remark}

\begin{document}

\title{The quantum supremum of the \texorpdfstring{$I_{3322}$}{I3322}
Bell inequality is not attained in finite dimension}

\author{Jef Pauwels}
\affiliation{Department of Applied Physics, University of Geneva, Switzerland}
\affiliation{Constructor University, Bremen, Germany}

\begin{abstract}
In 2010, P\'al and V\'ertesi found a family of finite-dimensional
strategies for the $I_{3322}$ Bell inequality whose optimized values
appeared to converge as the local Hilbert-space dimension grew.  They
conjectured that this limit is the supremum over all finite-dimensional
quantum strategies, but that no finite-dimensional strategy attains it.  We
prove both claims.  The proof uses the symmetry of the Bell functional to associate every strategy with a finite matrix
of probabilities, one for each pair of spectral subspaces of Alice and Bob.
This matrix gives an upper bound on the Bell value, and finite-dimensional
strategies built from the repeating structure found by P\'al and V\'ertesi
approach it as the dimension grows.  If the bound were attained exactly in finite dimension,
the optimality conditions would then require a state that cannot be normalized.
Consequently, the set of finite-dimensional quantum correlations
is not closed in the $(3,3,2,2)$ scenario, the smallest Bell scenario where this can happen.  Moreover,
approaching the supremum requires unbounded local dimension.  The core of the proof was formalized in Lean~4.
\end{abstract}

\maketitle

\section{Introduction}
\label{sec:introduction}

Does every Bell inequality \cite{Bell,ReviewBell} have an optimal
quantum strategy?  For a fixed local Hilbert-space dimension the answer
is yes, because the strategy set is compact and the maximum is attained.  But
quantum theory imposes no bound on the dimension: optimizing over all
finite dimensions means optimizing over an increasing union of compact
correlation sets, and that union need not be closed.  Strategies of increasing
dimension may then approach a value that no finite-dimensional
strategy reaches.

Slofstra proved that the set $\mathcal C_q$ of finite-dimensional
tensor-product correlations is not closed
\cite{Slofstra}.  A later construction exhibited nonclosure with five binary
measurements per party~\cite{DPP}, subsequently reduced to four~\cite{Beigi}.
Relatedly, Coladangelo and Stark explicitly constructed a correlation that intrinsically
requires infinite-dimensional entanglement \cite{ColadangeloStark}.
The later $\mathrm{MIP}^*=\mathrm{RE}$ result separated the closure
of finite-dimensional correlations from the commuting-operator set,
resolving Tsirelson's problem
\cite{SlofstraTsirelson,ScholzWerner,Fritz,JNPPSW} in the negative
\cite{MIPstarRE}.  However, these results use specially
constructed correlations or nonlocal games; they do not locate
nonclosure at a familiar Bell inequality in a small measurement
scenario.

Beigi observed that, if either party has only two binary measurements,
Jordan's block decomposition reduces every finite-dimensional
correlation to a convex combination of qubit--qubit correlations, and
the resulting set is closed \cite{Beigi}.
Thus $(3,3,2,2)$ (three settings per party, two outcomes each) is
the smallest binary-outcome scenario in which nonclosure is possible.

The natural inequality at this boundary is $I_{3322}$. It is the only facet class of the local
polytope in this scenario beyond positivity constraints and
relabelings or liftings \cite{PironioLifting} of CHSH
\cite{CHSH,Froissart,Sliwa2003,PitowskySvozil,CollinsGisin}.  It
detects some mixed two-qubit states that satisfy every CHSH inequality
\cite{CollinsGisin}, and its violation is sensitive to dimension:
earlier numerical optimization gave a two-qubit maximum of $1/4$
\cite{CollinsGisin,PalVertesi2008,GigenaKaniewski}, and
Appendix~\ref{app:qubit-bound} proves this value exactly.  P\'al and
V\'ertesi found higher-dimensional, nonmaximally entangled states with
values above it \cite{PV}; subsequent work further studied its
dimension dependence
\cite{MoroderEtAl,NavascuesDeLaTorreVertesi,VidickWehner}.
The inequality was implemented experimentally in 2005
\cite{AltepeterEtAl} and has since been tested with later photon
sources \cite{PomaricoEtAl}, studied in connection with detection
efficiency and device-independent cryptography
\cite{BrunnerEtAlDetection,SuDIQKD}, and used repeatedly as a benchmark
for methods that bound quantum correlations
\cite{LiangDoherty,GigenaKaniewski,BernardsGuhne,Mortimer2025,FloraEtAl2026}.

P\'al and V\'ertesi (PV) found a family of $I_{3322}$ strategies
with a simple repeating pattern \cite{PV}: on a state with Schmidt
coefficients $\lambda_1,\lambda_2,\ldots$, Alice's first two
measurements act within two-dimensional blocks of consecutive Schmidt
vectors, Bob's within blocks offset by one vector, and the third
measurements use the complementary pairings.  Optimizing this family in
increasing dimension produced values approaching
\[
 0.25087538\ldots.
\]
They conjectured both that this is the supremum over all
finite-dimensional strategies and that no finite-dimensional strategy
attains it.  The $I_{3322}$ attainment question was highlighted in
Ref.~\cite{DeltaGame}, which also notes the general connection between
closedness and attainment.
Coladangelo and Stark later
restated the second half in a stronger form, conjecturing that any
correlation maximally violating $I_{3322}$ requires
infinite-dimensional entanglement \cite{ColadangeloStark}.  Explicit bounds from
the Navascu\'es--Pironio--Ac\'in hierarchy supported these claims
but did not prove either statement \cite{NPA1,NPA2,AraujoEtAl2026}.

We prove both parts of P\'al and V\'ertesi's conjecture, as statements
about finite-dimensional strategies.  The supremum over all
finite-dimensional quantum strategies equals the supremum of the finite
P\'al--V\'ertesi family, and the value of every individual
finite-dimensional strategy lies strictly below it.  This proves that no
finite-dimensional correlation maximally violates $I_{3322}$; whether an
infinite-dimensional tensor-product correlation attains the supremum remains
open (Sec.~\ref{sec:consequences}).
It follows that $\mathcal C_q$ is not closed already in
the $(3,3,2,2)$ scenario, and that approaching the supremum requires
unbounded local dimension, in the sense of dimension witnesses
\cite{DimWitness,VertesiPalDimension,NV}.

Sections~\ref{sec:problem}--\ref{sec:strictness} give the proof,
Sec.~\ref{sec:consequences} develops its consequences and open questions,
and Appendix~\ref{app:lean} describes the accompanying Lean~4 development.

\section{The Bell functional and the P\'al--V\'ertesi family}
\label{sec:problem}

\subsection{The Bell functional}
\label{sec:bell-functional}

A finite-dimensional quantum strategy $\mathsf S$ consists of a
state on $\mathbb C^{d_A}\otimes\mathbb C^{d_B}$, for arbitrary
finite dimensions $d_A$ and $d_B$, together with effects
$0\leq A_x\leq\1$ and $0\leq B_y\leq\1$, one for each $x$ and
each $y$ in $\{1,2,3\}$.  Outcomes are labeled $0$ and $1$, and
$A_x$ and $B_y$ are the effects of outcome $1$, so that
$\langle A_xB_y\rangle=P(11|xy)$ and
$\langle A_x\rangle=P_A(1|x)$.  Our normalization is
\begin{align}
I_{3322}={}&-\langle A_2\rangle-\langle B_1\rangle
-2\langle B_2\rangle \nonumber\\
&+\langle A_1B_1\rangle+\langle A_1B_2\rangle
+\langle A_2B_1\rangle+\langle A_2B_2\rangle \nonumber\\
&-\langle A_1B_3\rangle+\langle A_2B_3\rangle
-\langle A_3B_1\rangle+\langle A_3B_2\rangle.
\label{eq:i3322}
\end{align}
Its local bound is zero~\cite{CollinsGisin}.  Let
\begin{equation}
\begin{aligned}
 \Istar:=\sup\bigl\{I_{3322}(\mathsf S):\;&
 \mathsf S\text{ is a finite-dimensional}\\[-2pt]
 &\text{quantum strategy}\bigr\}.
\end{aligned}
\label{eq:istar}
\end{equation}
The dimensions are unrestricted; the question is whether the supremum
is attained.

\subsection{The P\'al--V\'ertesi family}
\label{sec:pv-family}
We first recall the P\'al--V\'ertesi construction~\cite{PV}.
 In a Schmidt basis, take
\begin{equation}
 |\psi_n\rangle=\sum_{i=1}^n\lambda_i
 |i\rangle_A|i\rangle_B,
 \qquad \lambda_i\geq0,
 \qquad \sum_{i=1}^n\lambda_i^2=1.
\label{eq:pv-state}
\end{equation}
The $\lambda_i$ are the Schmidt coefficients of the state.  Alice's
first two measurements act on the pairs
$(2,3),(4,5),\ldots$ of Schmidt basis vectors, while Bob's act on the
offset pairs $(1,2),(3,4),\ldots$.  Their third measurements use the
complementary pairings.  Each two-dimensional block is characterized by
$|c|\in[0,1]$: on that block, half the sum of the party's first two
measurements, written as $\pm1$-valued observables, has eigenvalues
$\pm|c|$.  Below we use the signed eigenvalue $c\in[-1,1]$ as the
spectral label.  The magnitude $|c|$ records how much the two measurements
agree: at $|c|=1$ they coincide on the block, at $c=0$ one is the other
with its outcomes flipped, and intermediate values interpolate.  Within the
block, the two projectors differ only
in the sign of their off-diagonal entries, whose magnitude is
$\sqrt{1-c^2}/2$.  
Figure~\ref{fig:pv-family} shows the complete pairing pattern.

\begin{figure*}[t]
\centering
\begin{tikzpicture}[
  v3 every figure,
  x=1cm,
  y=1cm,
  pv basis/.style={
    draw=black!48,
    rounded corners=1pt,
    fill=white,
    minimum width=5.8mm,
    minimum height=5.4mm,
    inner sep=1pt
  }
]
  \node[v3 panel] at (0,4.10)
    {Offset pairings in the Schmidt basis};

  \node[v3 annotation, anchor=east, text=black] at (1.42,2.90)
    {Alice: $A_1,A_2$};
  \node[v3 annotation, anchor=east, text=black] at (1.42,1.62)
    {Bob: $B_1,B_2$};

  \node[pv basis] (alice1) at (1.90,2.90) {$\lvert1\rangle$};
  \node[pv basis] (alice2) at (2.98,2.90) {$\lvert2\rangle$};
  \node[pv basis] (alice3) at (4.06,2.90) {$\lvert3\rangle$};
  \node[pv basis] (alice4) at (5.14,2.90) {$\lvert4\rangle$};
  \node[pv basis] (alice5) at (6.22,2.90) {$\lvert5\rangle$};
  \node[pv basis] (alice6) at (7.30,2.90) {$\lvert6\rangle$};
  \node[pv basis] (alice7) at (8.38,2.90) {$\lvert7\rangle$};
  \node[pv basis] (bob1) at (1.90,1.62) {$\lvert1\rangle$};
  \node[pv basis] (bob2) at (2.98,1.62) {$\lvert2\rangle$};
  \node[pv basis] (bob3) at (4.06,1.62) {$\lvert3\rangle$};
  \node[pv basis] (bob4) at (5.14,1.62) {$\lvert4\rangle$};
  \node[pv basis] (bob5) at (6.22,1.62) {$\lvert5\rangle$};
  \node[pv basis] (bob6) at (7.30,1.62) {$\lvert6\rangle$};
  \node[pv basis] (bob7) at (8.38,1.62) {$\lvert7\rangle$};

  \node[draw=vthreeexact, very thick, rounded corners=3pt,
        fit=(alice1), inner xsep=3pt, inner ysep=3pt] (aone) {};
  \node[draw=vthreeexact, very thick, rounded corners=3pt,
        fit=(alice2)(alice3), inner xsep=3pt, inner ysep=3pt] (a23) {};
  \node[draw=vthreeexact, very thick, rounded corners=3pt,
        fit=(alice4)(alice5), inner xsep=3pt, inner ysep=3pt] (a45) {};
  \node[draw=vthreeexact, very thick, rounded corners=3pt,
        fit=(alice6)(alice7), inner xsep=3pt, inner ysep=3pt] (a67) {};
  \node[v3 annotation, text=vthreeexact, anchor=south]
    at ([yshift=1pt]aone.north) {$c_0=1$};
  \node[v3 annotation, text=vthreeexact, anchor=south]
    at ([yshift=1pt]a23.north) {$c_2$};
  \node[v3 annotation, text=vthreeexact, anchor=south]
    at ([yshift=1pt]a45.north) {$c_4$};
  \node[v3 annotation, text=vthreeexact, anchor=south]
    at ([yshift=1pt]a67.north) {$c_6$};

  \node[draw=vthreepath, very thick, rounded corners=3pt,
        fit=(bob1)(bob2), inner xsep=3pt, inner ysep=3pt] (b12) {};
  \node[draw=vthreepath, very thick, rounded corners=3pt,
        fit=(bob3)(bob4), inner xsep=3pt, inner ysep=3pt] (b34) {};
  \node[draw=vthreepath, very thick, rounded corners=3pt,
        fit=(bob5)(bob6), inner xsep=3pt, inner ysep=3pt] (b56) {};
  \node[draw=vthreepath, very thick, rounded corners=3pt,
        fit=(bob7), inner xsep=3pt, inner ysep=3pt] (bzero) {};
  \node[v3 annotation, text=vthreepath, anchor=north]
    at ([yshift=-1pt]b12.south) {$c_1$};
  \node[v3 annotation, text=vthreepath, anchor=north]
    at ([yshift=-1pt]b34.south) {$c_3$};
  \node[v3 annotation, text=vthreepath, anchor=north]
    at ([yshift=-1pt]b56.south) {$c_5$};
  \node[v3 annotation, text=vthreepath, anchor=north]
    at ([yshift=-1pt]bzero.south) {$c_n=-1$};

\end{tikzpicture}
 \caption{\label{fig:pv-family}
The P\'al--V\'ertesi construction.  Alice's and Bob's first two
measurements use two-dimensional blocks offset by one Schmidt basis
vector; the third measurements use the complementary pairings.  A finite
PV strategy is specified by its Schmidt coefficients $\lambda_i$ and
the spectral labels $c_i$ appearing in
Eq.~\eqref{eq:pv-value}.}
\end{figure*}
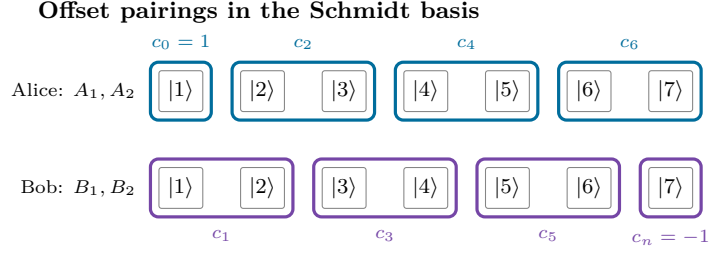

Define
\begin{equation}
 s(c):=\sqrt{1-c^2},
 \qquad
 d(a,b):=ab+\frac{a-b}{2}-1.
\label{eq:ds-main}
\end{equation}
Let $c_1,\ldots,c_{n-1}\in[-1,1]$ be the internal spectral labels, and
set $c_0=1$, $c_n=-1$ for the two ends, where the unpaired basis
vectors form one-dimensional (scalar) blocks.  Direct substitution
gives the Bell value
\begin{equation}
\mathcal P_n(c,\lambda)
=\frac{
 \displaystyle\sum_{i=1}^n d(c_{i-1},c_i)\lambda_i^2
 +\displaystyle\sum_{i=1}^{n-1}s(c_i)\lambda_i\lambda_{i+1}}
 {\displaystyle\sum_{i=1}^n\lambda_i^2}.
\label{eq:pv-value}
\end{equation}
The denominator lets us work with unnormalized Schmidt coefficients:
dividing the $\lambda_i$ by $\sqrt{\sum_j\lambda_j^2}$ gives the
Schmidt coefficients of the physical state.
For odd $n$, the block projectors of Ref.~\cite{PV} realize arbitrary
$c_1,\ldots,c_{n-1}$; the corresponding Bell operator in the Schmidt
basis has diagonal entries $d(c_{i-1},c_i)$ and neighboring entries
$s(c_i)/2$, so its normalized quadratic expression is
Eq.~\eqref{eq:pv-value}.  If
$n$ is even, adjoining
$c_{n+1}=-1$ and $\lambda_{n+1}=0$ embeds the same chain in the
next odd local dimension without changing its value.
The same data can be laid out on a line:
\begin{equation}
 c_0\xleftrightarrow{\ \lambda_1\ }c_1
 \xleftrightarrow{\ \lambda_2\ }\cdots
 \xleftrightarrow{\ \lambda_n\ }c_n.
\label{eq:pv-list-picture}
\end{equation}
Thus $n$ is the length of the Schmidt-coefficient vector and equals the
local dimension in the odd-$n$ realization.  An even-length chain is
realized in local dimension $n+1$ by the zero padding above, while its
Schmidt rank remains at most $n$.
We call this alternating list of spectral labels and Schmidt coefficients a
\emph{PV chain}.  Each Schmidt coefficient $\lambda_i$ sits between the labels
$(c_{i-1},c_i)$: the $d$ term is a diagonal contribution, while
the $s$ term couples neighboring Schmidt coefficients.
Two elementary features of Eq.~\eqref{eq:pv-value} will matter.  First,
$d\leq0$ on $[-1,1]^2$, with equality only at $(1,1)$ and
$(-1,-1)$, so every positive contribution comes from the
nearest-neighbor terms.  Second, if every internal label equals
the same $c$, then
$2\lambda_i\lambda_{i+1}\leq\lambda_i^2+\lambda_{i+1}^2$ gives a
interior coefficient $d(c,c)+s(c)=s(c)-s(c)^2\leq1/4$; the two endpoint
coefficients are even smaller.  Such a chain therefore has value at
most $1/4$, the exact two-qubit threshold proved in
Appendix~\ref{app:qubit-bound}.  Exceeding $1/4$ requires the spectral
labels to vary along the chain.

We define the PV supremum directly from this family:
\begin{equation}
\betaPV:=
\sup_{\substack{n\geq1,\ c_0=1,\ c_n=-1\\
 c_1,\ldots,c_{n-1}\in[-1,1]\\
 \lambda\in\mathbb R_{\geq0}^n\setminus\{0\}}}
\mathcal P_n(c,\lambda).
\label{eq:beta-pv}
\end{equation}
For fixed $n$, normalization makes the parameter domain compact, so
the maximum exists.  The supremum over all $n$ need not be attained
at any finite $n$.

\subsection{Results and proof overview}
\label{sec:results}

The conjecture has two logically distinct parts.

\begin{theorem}[The PV family gives the finite-dimensional supremum]
\label{thm:variational}
The finite-dimensional quantum supremum equals the supremum of the finite
PV family:
\begin{equation}
 \Istar=\betaPV.
\label{eq:pv-identity}
\end{equation}
\end{theorem}

Theorem~\ref{thm:variational} says that no construction, however
different from the PV family, can improve on the PV supremum.  The
proof is the string of inequalities
\[
 I_{3322}(\mathsf S)
 \;\leq\;
 \Phi(\theta)
 \;\leq\;
 \betaPV
 \;\leq\;
 \Istar,
\]
valid for every finite-dimensional strategy $\mathsf S$; taking the
supremum over $\mathsf S$ turns the whole line into equalities.

Here $\theta$ is the spectral-weight matrix constructed in
Sec.~\ref{sec:strategy-table}: a joint probability distribution over spectral
labels associated with Alice's and Bob's first two measurements.  To obtain it, a direct-sum doubling replaces
any strategy by an equally valued symmetric one in which Bob's
measurements are outcome-flipped, conjugated copies of Alice's.  The third measurements,
which enter linearly, can then be optimized in closed form, and
Jordan's lemma reduces the remaining pair of measurements to finitely
many spectral subspaces.  One Cauchy--Schwarz estimate bounds the Bell
value by the explicit functional $\Phi(\theta)$ of their joint
weights (Lemma~\ref{lem:table-bound}).

Conversely, for every length $N$, these weights generate a finite
collection of PV chains whose weighted-average value
$\Phi_N(\theta)$ converges to $\Phi(\theta)$ as $N\to\infty$
(Sec.~\ref{sec:table-to-pv}).  Thus every $\Phi(\theta)$ is obtained in
the large-dimension limit as a weighted average of the Bell values of finite
PV strategies.
Since each such average is at most $\betaPV$, it follows that
$\Phi(\theta)\leq\betaPV$.

Finally, every finite PV chain is itself
a finite-dimensional strategy by the realization above.

\begin{theorem}[Finite-dimensional nonattainment]
\label{thm:nonattainment}
No finite-dimensional quantum strategy attains $\Istar$.
\end{theorem}

For nonattainment, the spectral-weight matrix is now viewed as an auxiliary
joint probability distribution.  If a finite-dimensional strategy were
optimal, the optimality conditions for this distribution would force an
associated construction with the same form as a PV strategy, but with a
coefficient sequence that cannot be normalized.  This contradiction proves
Theorem~\ref{thm:nonattainment} (Sec.~\ref{sec:strictness}).

The theorems have two immediate consequences.

\begin{corollary}[Nonclosure]
\label{cor:nonclosure}
The set of finite-dimensional quantum correlations is not closed in the
$(3,3,2,2)$ scenario:
\begin{equation}
 \mathcal C_q(3,3,2,2)
 \subsetneq
 \mathcal C_{qa}(3,3,2,2),
\label{eq:nonclosure-main}
\end{equation}
where $\mathcal C_{qa}:=\overline{\mathcal C_q}$ is the closure of
$\mathcal C_q$.
\end{corollary}

\begin{corollary}[Dimension witness]
\label{cor:dimension-witness}
Let $\beta_k$ be the maximum of $I_{3322}$ over strategies whose
local dimensions are both at most $k$.  Then
\begin{equation}
 \beta_k<\Istar
 \qquad\text{for every finite }k;
\label{eq:dimension-gap}
\end{equation}
any sequence of strategies whose values approach $\Istar$ has
unbounded local dimension.
\end{corollary}

\section{From strategies to spectral weights}
\label{sec:strategy-table}

We now prove the first half of Theorem~\ref{thm:variational},
reducing an arbitrary strategy to a finite joint distribution of
spectral labels whose upper-bound functional controls the Bell value
(Lemma~\ref{lem:table-bound}).

\subsection{Dichotomic observables and the symmetric form}
\label{sec:symmetric-form}

Purification and a finite-dimensional Naimark dilation let us assume a
pure state and projective binary measurements without changing any
probability.
We pass from the $0/1$-valued effects to dichotomic
observables---Hermitian
operators squaring to the identity, which describe the same measurements
recorded with outcomes $\pm1$:
\[
 a_x=2A_x-\1,
 \qquad
 b_y=2B_y-\1.
\]
Alice's first two measurements enter through their sum and their
difference: the sum supplies the spectral labels used below; the
difference couples linearly to the third measurement.  Put
\[
 p=a_1+a_2,
 \qquad
 r=a_2-a_1.
\]
Then
\begin{equation}
 p^2+r^2=4\1,
 \qquad
 pr+rp=0.
\label{eq:pair-algebra}
\end{equation}
In these variables,
\begin{equation}
4(1+I_{3322})
=\langle p\rangle-\langle q\rangle+\langle pq\rangle
 +\langle rb_3\rangle+\langle a_3\tau\rangle,
\label{eq:F-reflections}
\end{equation}
where $q=b_1+b_2$ and $\tau=b_2-b_1$ are Bob's analogues of $p$
and $r$.
Up to local outcome relabelings, Eq.~\eqref{eq:F-reflections} is the
standard party-permutation-invariant form of $I_{3322}$
\cite{Sliwa2003,BrunnerGisin2008}.  More precisely, flipping the
outcomes of $a_2$, $b_1$, and $b_3$ converts it into the
sum-and-difference factorization written in Eq.~(2) of
Ref.~\cite{BernardsGuhne}.

After padding the smaller local space, encode the state
in a matrix:
$|\psi\rangle=\operatorname{vec}(D):=\sum_{ij}D_{ij}\,|i\rangle_A|j\rangle_B$.
For a state written in its Schmidt basis, $D$ is diagonal and
nonnegative, with the Schmidt coefficients on its diagonal, and
$\Tr D^2=1$.  This encoding gives the identity
$\langle\psi|X\otimes Y|\psi\rangle=\Tr(DXDY^{\mathsf T})$, which
turns every expectation value into a trace on a single local space.

We now use the symmetry of $I_{3322}$ to reduce the number of independent measurements.
A standard direct-sum symmetrization with a classical flag
\cite{MoroderEtAl,HsuEtAl} doubles the local spaces and gives the
following exact form.

\begin{lemma}[Symmetric form]
\label{lem:symmetric-form}
Every finite-dimensional strategy has the same Bell value as one satisfying
\begin{equation}
\begin{gathered}
 |\psi\rangle=\operatorname{vec}(D),
 \qquad D\geq0,
 \qquad \Tr D^2=1,\\
 [D,W]=0,
 \qquad W^2=\1,
 \qquad b_i^{\mathsf T}=-Wa_iW,
\end{gathered}
\label{eq:symmetric-form}
\end{equation}
for a fixed Hermitian unitary $W$.
\end{lemma}

\begin{proof}
Start from the pure, projective strategy described above, written in a
Schmidt basis as $|\psi\rangle=\operatorname{vec}(D)$.  On doubled
local spaces define
\begin{align*}
 \widetilde D&=\frac{D\oplus D}{\sqrt2},
 &W&=\begin{pmatrix}0&\1\\ \1&0\end{pmatrix},\\
 \widetilde a_i&=a_i\oplus(-b_i^{\mathsf T}),
 &\widetilde b_i&=b_i\oplus(-a_i^{\mathsf T}).
\end{align*}
The second block is the complex conjugate of the copy obtained by the
combined operation that swaps the parties and flips every outcome.
Complex conjugation preserves all
probabilities because the Schmidt matrix $D$ is real, and
$I_{3322}$ is invariant under this combined operation.  Thus
the two equally weighted blocks have the same Bell value.  The doubled
measurements remain $\pm1$-valued observables, and direct block
multiplication gives
\[
 [\widetilde D,W]=0,
 \qquad W^2=\1,
 \qquad \widetilde b_i^{\mathsf T}=-W\widetilde a_iW.
\]
Dropping the tildes proves the claim.
\end{proof}

Thus a symmetric strategy is specified by $D$, $W$, and Alice's
three measurements; the transpose in Eq.~\eqref{eq:symmetric-form}
comes only from the coefficient-matrix convention for
$\operatorname{vec}$.

\subsection{Eliminating the third measurements}
\label{sec:third-measurements}

Alice's remaining independent third measurement can now be optimized exactly;
Bob's is fixed by symmetry.  Using the trace
identity above, cyclicity of the trace, and $[D,W]=0$, the three
parts of Eq.~\eqref{eq:F-reflections} become
\begin{align}
 \langle p\rangle-\langle q\rangle&=2\Tr(D^2p),\nonumber\\
 \langle pq\rangle&=-\Tr(DpDWpW),\nonumber\\
 \langle rb_3+a_3\tau\rangle&=-2\Tr(a_3WDrDW).
\label{eq:reduced-trace-terms}
\end{align}
In the last line, the symmetry $b_y^{\mathsf T}=-Wa_yW$ turns the
$\tau$ term into a second copy of the $r$ term, so only a
single trace appears.
We use the standard identity
$\max_{a_3^2=\1}[-2\Tr(a_3H)]=2\lVert H\rVert_1$.  Applied to
$H=WDrDW$, and using that $W$ is unitary, this gives
$2\lVert DrD\rVert_1$.  The analogous positive-eigenspace
optimization of the third measurements was used for $I_{3322}$ and
a maximally entangled state in Ref.~\cite{VidickWehner}; here the same
optimization is weighted by the Schmidt matrix $D$.  Consequently,
\begin{align}
\Istar=\sup_{D,W,a_1,a_2}\bigg\{-1+\frac14\Big[
&2\Tr(D^2p)-\Tr(DpDWpW)\nonumber\\
&+2\lVert DrD\rVert_1\Big]\bigg\},
\label{eq:reduced-I}
\end{align}
under the conditions in Eq.~\eqref{eq:symmetric-form} and
$a_1^2=a_2^2=\1$.
The two sides of Eq.~\eqref{eq:reduced-I} coincide because every
tuple $(D,W,a_1,a_2)$ defines a quantum strategy---take the optimal
$a_3$ and set $b_i=(-Wa_iW)^{\mathsf T}$---and because replacing
$a_3$ by its optimum never lowers a strategy's value.

\subsection{The spectral-weight matrix}
\label{sec:coupling-table}

After Eq.~\eqref{eq:reduced-I}, a strategy is fully specified by three
pieces of data: the state matrix $D$, the unitary $W$, and two independent
binary projective measurements.  Jordan's lemma
\cite{Jordan1875,PironioEtAl2009} applies to this pair, and we
now retain only finitely many scalar weights.  Let $E_c$ be the
spectral projector of $p=a_1+a_2$ onto the eigenvalue $2c$, where
the finitely many labels $c$ lie in $[-1,1]$.  We close this
label set under $c\mapsto-c$ by assigning a zero projector when
an eigenvalue is absent.  The eigenspaces associated with $c$ and
$-c$---the ranges of $E_c$ and $E_{-c}$---together decompose
into Jordan blocks with parameter $|c|$ (with the range counted only
once when $c=0$).  Thus $|c|$ is the Jordan-block parameter, while
$c$ is the signed spectral label used in the PV family.
Equation~\eqref{eq:pair-algebra} implies
\begin{equation}
 rE_c=E_{-c}r,
 \qquad
 r^2E_c=4s(c)^2E_c.
\label{eq:mirror-sectors}
\end{equation}
Thus the difference $r=a_2-a_1$ maps the $c$-eigenspace to the
paired $-c$-eigenspace, with magnitude $2s(c)$.

The same paired structure appears on Bob's side.  Let
$G_{c'}$ be the spectral projector of $q=b_1+b_2$ onto the
eigenvalue $2c'$; the symmetry $b_i^{\mathsf T}=-Wa_iW$ of
Eq.~\eqref{eq:symmetric-form} gives $q^{\mathsf T}=-WpW$ and hence
$G_{c'}^{\mathsf T}=WE_{-c'}W$.  Define the \emph{spectral-weight
matrix}
\begin{equation}
 \theta_{cc'}:=\langle\psi|\,E_c\otimes G_{c'}\,|\psi\rangle,
\label{eq:theta}
\end{equation}
which is the joint probability of spectral label $c$ for Alice
and $c'$ for Bob.  This is an auxiliary distribution over spectral
subspaces.  Writing $J=DW$,
so that $J=J^\dagger$ and $J^2=D^2$, the same quantity takes the
form we will use:
\begin{equation}
 \theta_{cc'}=\Tr(JE_cJE_{-c'})=\lVert E_cJE_{-c'}\rVert_2^2\geq0,
\label{eq:theta-hs}
\end{equation}
where the sign reversal in $E_{-c'}$ follows from the direct-sum swap
symmetry.  The row and column marginals are the corresponding spectral
weights,
\begin{align}
 R_c&:=\sum_{c'}\theta_{cc'}=\Tr(E_cD^2),\nonumber\\
 C_c&:=\sum_{c'}\theta_{c'c}=\Tr(E_{-c}D^2),
\label{eq:row-column}
\end{align}
and the total weight is $\sum_{c,c'}\theta_{cc'}=1$.

Direct expansion in the spectral projections of $p$, using
Eq.~\eqref{eq:theta-hs} and $J^2=D^2$, gives
\begin{align}
 \sum_{c,c'}c\,\theta_{cc'}&=\frac12\Tr(D^2p),\nonumber\\
 \sum_{c,c'}c'\,\theta_{cc'}&=-\frac12\Tr(D^2p),\nonumber\\
 \sum_{c,c'}cc'\,\theta_{cc'}&=-\frac14\Tr(pJpJ).
\label{eq:table-moments}
\end{align}
The minus sign in the second moment restates
$\langle q\rangle=-\langle p\rangle$, which the swap symmetry
enforces.  Substituting the definition of $d$ shows that the part of
Eq.~\eqref{eq:reduced-I} not containing $r$ is exactly
\begin{equation}
 \sum_{c,c'}d(c,c')\theta_{cc'}
 =-1+\frac14\left[
 2\Tr(D^2p)-\Tr(pJpJ)
 \right].
\label{eq:exact-table-term}
\end{equation}
Here $\Tr(pJpJ)=\Tr(DpDWpW)$ by cyclicity and $[D,W]=0$.

The following estimate is the only step in this reduction that can enlarge
the value:
\begin{align}
 \lVert DrD\rVert_1
 &\leq\sum_c\lVert DE_{-c}rE_cD\rVert_1\nonumber\\
 &\leq\sum_c\lVert DE_{-c}\rVert_2
                  \lVert rE_cD\rVert_2\nonumber\\
 &=2\sum_cs(c)\sqrt{R_cC_c}.
\label{eq:trace-bound}
\end{align}
Here the first step decomposes $r=\sum_cE_{-c}rE_c$, which follows
from Eq.~\eqref{eq:mirror-sectors}, and applies the triangle
inequality; the second step is the Cauchy--Schwarz inequality
$\lVert AB\rVert_1\leq\lVert A\rVert_2\lVert B\rVert_2$.  Each
paired-eigenspace contribution is therefore bounded by $s(c)$ times
the geometric mean of the corresponding marginals.

For any finite set $X\subset[-1,1]$ and any nonnegative matrix
$\theta=(\theta_{cc'})_{c,c'\in X}$ of total weight one, define
\begin{equation}
 \Phi(\theta):=
 \sum_{c,c'}d(c,c')\theta_{cc'}
 +\sum_cs(c)\sqrt{R_cC_c}.
\label{eq:Phi}
\end{equation}
The first sum is exact; the second is the upper bound supplied by
Eq.~\eqref{eq:trace-bound}.  We therefore refer to $\Phi$ simply as
the upper-bound functional.

\begin{lemma}[Spectral-weight bound]
\label{lem:table-bound}
Every finite-dimensional quantum strategy $\mathsf S$ yields a finite
spectral-weight matrix $\theta$ with
\begin{equation}
 I_{3322}(\mathsf S)\leq\Phi(\theta).
\label{eq:strategy-table}
\end{equation}
\end{lemma}

\begin{proof}
Bring $\mathsf S$ to symmetric form by
Lemma~\ref{lem:symmetric-form} and combine
Eqs.~\eqref{eq:reduced-I}, \eqref{eq:exact-table-term},
and~\eqref{eq:trace-bound}.
\end{proof}

\section{From spectral weights to PV chains}
\label{sec:table-to-pv}

Section~\ref{sec:strategy-table} established
$I_{3322}(\mathsf S)\leq\Phi(\theta)$.  To finish the proof of
Theorem~\ref{thm:variational}, it remains only to show that
$\Phi(\theta)\leq\betaPV$
(Lemma~\ref{lem:finite-sequence-averaging} below).
Recall that a PV chain is the alternating list
$c_0,\lambda_1,c_1,\ldots,\lambda_n,c_n$ of spectral labels and
Schmidt coefficients pictured in Eq.~\eqref{eq:pv-list-picture}, whose value
$\mathcal P_n(c,\lambda)$ is given by Eq.~\eqref{eq:pv-value}.

The construction below is purely combinatorial: it works for any
nonnegative matrix $\theta$ of total weight one, indexed by a finite
subset of $[-1,1]$, whether or not it is the spectral-weight matrix
of a strategy.  We temporarily omit the fixed endpoint labels $1$ and $-1$;
they are restored by zero padding at the end of the proof.  Fix a length
$N$.  We construct a finite collection of PV chains $\gamma$, written
with spectral labels $c_0^\gamma,\ldots,c_N^\gamma$ and unnormalized Schmidt coefficients
$\lambda_1^\gamma,\ldots,\lambda_N^\gamma$.  At every position $i$,
the squared coefficients associated with an ordered pair of labels
$(a,b)$ will sum to the matrix entry $\theta_{ab}$:
\begin{equation}
 \sum_{\substack{\gamma:\,
 (c_{i-1}^\gamma,c_i^\gamma)=(a,b)}}
 (\lambda_i^\gamma)^2=\theta_{ab}
 \qquad\text{for every }a,b.
\label{eq:target-diagonal-total}
\end{equation}
Thus every position contributes
$\sum_{a,b}d(a,b)\theta_{ab}$ after summing over the collection, and
summing Eq.~\eqref{eq:target-diagonal-total} over $a,b$ shows that
$\sum_\gamma(\lambda_i^\gamma)^2=1$ at every position.

Reproducing the matrix entries does not yet reproduce the second term in
$\Phi(\theta)$, which contains products of neighboring Schmidt
coefficients.  Two neighboring coefficients
\[
 a\xleftrightarrow{\ \lambda_i\ }c
 \xleftrightarrow{\ \lambda_{i+1}\ }b,
\]
contribute $s(c)\lambda_i\lambda_{i+1}$.  We shall also arrange that
\begin{equation}
 \sum_{\gamma:\,c_i^\gamma=c}
 \lambda_i^\gamma\lambda_{i+1}^\gamma
 =\sqrt{C_cR_c}
 \qquad\text{for every }c\text{ and }1\leq i<N,
\label{eq:target-neighbor-total}
\end{equation}
the coefficient of $s(c)$ in $\Phi$.  Here
$C_c=\sum_a\theta_{ac}$ and $R_c=\sum_b\theta_{cb}$ are the column
and row marginals.  We call Eqs.~\eqref{eq:target-diagonal-total}
and~\eqref{eq:target-neighbor-total} the \emph{matching conditions}:
position by position, the collection must reproduce the diagonal
weights and the nearest-neighbor couplings that enter
$\Phi(\theta)$.

A chain with $N$ Schmidt coefficients has $N$ diagonal contributions
but only $N-1$ neighboring-coefficient contributions.
Consequently the averaged value of the constructed chains will be
\[
 \Phi_N(\theta):=
 \sum_{a,b}d(a,b)\theta_{ab}
 +\frac{N-1}{N}\sum_cs(c)\sqrt{R_cC_c}.
\]
The finite-size factor tends to one as $N\to\infty$.

\begin{lemma}[Approximation by finite PV chains]
\label{lem:finite-sequence-averaging}
Every nonnegative matrix of total weight one, indexed by a finite subset of
$[-1,1]$, satisfies
\begin{equation}
 \Phi(\theta)\leq\betaPV.
\label{eq:table-pv-bound}
\end{equation}
\end{lemma}

\begin{proof}
Assume first that every marginal is positive, fix $N$, and let
$\mathcal X$ be the finite label set.  For every sequence
$\gamma=(c_0^\gamma,\ldots,c_N^\gamma)\in\mathcal X^{N+1}$, define
the squared coefficient at position $i$, $1\leq i\leq N$, by
\begin{align}
 (\lambda_i^\gamma)^2
 :={}&\theta_{c_{i-1}^\gamma c_i^\gamma}
 \prod_{j=1}^{i-1}
 \frac{\theta_{c_{j-1}^\gamma c_j^\gamma}}{C_{c_j^\gamma}}
 \prod_{j=i+1}^{N}
 \frac{\theta_{c_{j-1}^\gamma c_j^\gamma}}{R_{c_{j-1}^\gamma}}.
\label{eq:all-index-coefficients}
\end{align}

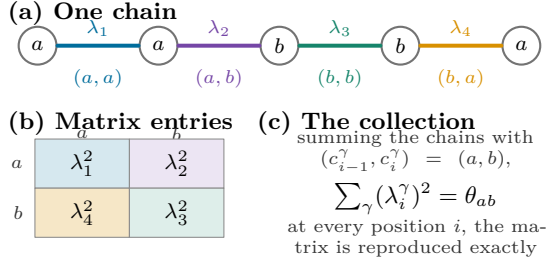
\begin{figure}[t]
\centering
\begin{tikzpicture}[
  v3 every figure,
  x=1cm,
  y=1cm,
  parameter node/.style={
    circle,
    draw=black!55,
    thick,
    fill=white,
    minimum size=5.0mm,
    inner sep=1pt
  },
  cell/.style={
    draw=black!45,
    minimum width=9.5mm,
    minimum height=6.5mm,
    inner sep=1pt
  }
]
  \node[v3 panel] at (0.05,0.12) {(a) One chain};
  \node[parameter node] (c0) at (0.45,-0.48) {$a$};
  \node[parameter node] (c1) at (2.05,-0.48) {$a$};
  \node[parameter node] (c2) at (3.65,-0.48) {$b$};
  \node[parameter node] (c3) at (5.25,-0.48) {$b$};
  \node[parameter node] (c4) at (6.85,-0.48) {$a$};
  \draw[vthreeexact, line width=1.5pt] (c0) -- (c1)
    node[midway,above,font=\scriptsize,text=vthreeexact] {$\lambda_1$};
  \draw[vthreepath, line width=1.5pt] (c1) -- (c2)
    node[midway,above,font=\scriptsize,text=vthreepath] {$\lambda_2$};
  \draw[vthreegreen, line width=1.5pt] (c2) -- (c3)
    node[midway,above,font=\scriptsize,text=vthreegreen] {$\lambda_3$};
  \draw[vthreerelaxed, line width=1.5pt] (c3) -- (c4)
    node[midway,above,font=\scriptsize,text=vthreerelaxed] {$\lambda_4$};
  \node[v3 annotation,text=vthreeexact] at (1.25,-0.91) {$(a,a)$};
  \node[v3 annotation,text=vthreepath] at (2.85,-0.91) {$(a,b)$};
  \node[v3 annotation,text=vthreegreen] at (4.45,-0.91) {$(b,b)$};
  \node[v3 annotation,text=vthreerelaxed] at (6.05,-0.91) {$(b,a)$};

  \node[v3 panel] at (0.05,-1.34) {(b) Matrix entries};
  \node[v3 annotation] at (1.04,-1.65) {$a$};
  \node[v3 annotation] at (2.29,-1.65) {$b$};
  \node[v3 annotation,anchor=east] at (0.37,-2.04) {$a$};
  \node[v3 annotation,anchor=east] at (0.37,-2.69) {$b$};
  \node[cell,minimum width=12.5mm,fill=vthreeexact!16]
    at (1.04,-2.04) {$\lambda_1^2$};
  \node[cell,minimum width=12.5mm,fill=vthreepath!14]
    at (2.29,-2.04) {$\lambda_2^2$};
  \node[cell,minimum width=12.5mm,fill=vthreerelaxed!22]
    at (1.04,-2.69) {$\lambda_4^2$};
  \node[cell,minimum width=12.5mm,fill=vthreegreen!14]
    at (2.29,-2.69) {$\lambda_3^2$};

  \node[v3 panel] at (3.35,-1.34) {(c) The collection};
  \node[v3 annotation,text width=4.2cm] at (5.45,-1.88)
    {summing the chains with $(c_{i-1}^\gamma,c_i^\gamma)=(a,b)$,};
  \node[font=\small] at (5.45,-2.50)
    {$\textstyle\sum_\gamma(\lambda_i^\gamma)^2=\theta_{ab}$};
  \node[v3 annotation,text width=4.2cm] at (5.45,-3.02)
    {at every position $i$, the matrix is reproduced exactly};
\end{tikzpicture}
 \caption{\label{fig:table-pv-correspondence}
How a collection of PV chains reproduces the matrix.  Panels (a) and (b) show this
two-label example,
for a chain $\gamma$, $(\lambda_i^\gamma)^2$ is assigned to the matrix entry
indexed by $(c_{i-1}^\gamma,c_i^\gamma)$; colors match positions and entries.  (c) Summed over the
collection, the squared coefficients at each position reproduce every entry of
$\theta$ exactly [Eq.~\eqref{eq:target-diagonal-total}].  Only
neighboring-coefficient
contributions carry $(N-1)/N$, since $N$ coefficients have $N-1$ neighboring
pairs.}
\end{figure}

Here $\lambda_i^\gamma$ is the nonnegative square root, and an empty
product equals one.  Direct summation of
Eq.~\eqref{eq:all-index-coefficients} gives both matching conditions,
Eqs.~\eqref{eq:target-diagonal-total}
and~\eqref{eq:target-neighbor-total}; the calculation is given in
Appendix~\ref{app:matching-identities}.  Thus, after discarding chains
whose coefficients all vanish, the normalized coefficients
$\lambda_i^\gamma/\sqrt{\sum_j(\lambda_j^\gamma)^2}$ define finite
states, which become PV strategies after the endpoint padding below.

For each remaining chain $\gamma$, let $\nu_\gamma$ and $S_\gamma$ be
the numerator and denominator of
Eq.~\eqref{eq:pv-value} evaluated on $\gamma$, so that its PV value is
$\nu_\gamma/S_\gamma$.  Summing the matching conditions,
Eqs.~\eqref{eq:target-diagonal-total}
and~\eqref{eq:target-neighbor-total},
 over all positions and adjacent pairs gives
\begin{align*}
 \sum_\gamma \nu_\gamma
 &=N\sum_{a,b}d(a,b)\theta_{ab}\\
 &\quad +(N-1)\sum_cs(c)\sqrt{R_cC_c},\\
 \sum_\gamma S_\gamma&=N.
\end{align*}
To meet the endpoint convention in Eq.~\eqref{eq:beta-pv}, replace the
spectral labels and unnormalized Schmidt coefficients of each chain by
 \[
  (1,c_0,\ldots,c_N,-1),
  \qquad
  (0,\lambda_1^\gamma,\ldots,\lambda_N^\gamma,0),
\]
respectively; all added terms vanish, and each padded chain is admissible in
the variational family defined by Eq.~\eqref{eq:beta-pv}.  Therefore
\begin{equation}
 \Phi_N(\theta)
 =\frac{\sum_\gamma \nu_\gamma}{\sum_\gamma S_\gamma}
 \leq\max_\gamma\frac{\nu_\gamma}{S_\gamma}
 \leq\betaPV.
\label{eq:collection-pv-bound}
\end{equation}
The ratio is the weighted average of the chains' PV values, with weights
$S_\gamma/\sum_{\gamma'}S_{\gamma'}$, and is therefore at most their
maximum.
Letting $N\to\infty$ removes the factor $(N-1)/N$ and gives
$\Phi(\theta)\leq\betaPV$.

If some marginal vanishes, let $\mu$ be the uniform matrix on the same
finite label set and put
$\theta^\varepsilon=(1-\varepsilon)\theta+\varepsilon\mu$.
Every marginal of $\theta^\varepsilon$ is positive, so the preceding
argument gives $\Phi(\theta^\varepsilon)\leq\betaPV$ for every
$0<\varepsilon\leq1$.  Continuity of $\Phi$ then gives
$\Phi(\theta)\leq\betaPV$ as $\varepsilon\downarrow0$.
\end{proof}

\begin{proof}[Proof of Theorem~\ref{thm:variational}]
Every finite-dimensional strategy produces a finite
spectral-weight matrix
satisfying
\[
 I_{3322}(\mathsf S)\leq\Phi(\theta)\leq\betaPV
\]
(Lemmas~\ref{lem:table-bound} and~\ref{lem:finite-sequence-averaging}).
Hence $\Istar\leq\betaPV$.  Conversely, the realization following
Eq.~\eqref{eq:pv-value} shows that every finite PV chain is a
finite-dimensional quantum strategy, so $\Istar\geq\betaPV$.  Thus
$\Istar=\betaPV$.
\end{proof}

\section{Why finite-dimensional strategies cannot attain the supremum}
\label{sec:strictness}

It remains to show that no finite-dimensional strategy reaches
$\Istar=\betaPV$.  Section~\ref{sec:table-to-pv} constructed, for every
finite spectral-weight matrix $\theta$ and every length $N$, a finite collection of
PV chains whose weighted average has value $\Phi_N(\theta)$, with
$\Phi_N(\theta)\to\Phi(\theta)$ as $N\to\infty$.  Hence, if
$\Phi(\theta)=\betaPV$, some finite PV chains have values approaching
$\betaPV$.

This does not produce a single PV chain attaining $\betaPV$: the
chain selected from the collection may change with $N$, and its
dimension grows.  Moreover, proving only that no finite PV chain
attains $\betaPV$ would not exclude a different finite-dimensional
strategy from attaining the same value.  Equality of suprema does not
imply equality of maximizers.

Assume, for contradiction, that a finite-dimensional strategy attains
the supremum.  Its finite spectral-weight matrix---the joint
distribution of the two parties' spectral labels,
Eq.~\eqref{eq:theta}---then satisfies
\[
 \betaPV=I_{3322}(\mathsf S)\leq\Phi(\theta)\leq\betaPV,
\]
and hence $\Phi(\theta)=\betaPV$.  We now use this equality for the
finite matrix.  Combining maximality of $\Phi$ with the identity that
expresses it as a $\theta$-weighted average shows that every positive
entry satisfies an exact equality.  Lemma~\ref{lem:equality-chain} then
selects one possible sequence of positive entries whose indices match:
\[
 \theta_{c_0c_1}>0,\qquad
 \theta_{c_1c_2}>0,\qquad\ldots .
\]
This list need not contain every positive entry and is not a
representation of the whole matrix, or of the original quantum
strategy, as one PV chain.  From the row and column sums of the full
matrix, the lemma defines new auxiliary coefficients $u_i$ along this
selected list.  They play exactly the role of the $\lambda_i$ in the
PV expression, but the different symbol emphasizes that they are newly
constructed from $\theta$, rather than Schmidt coefficients of the
original state.  The finite matrix initially forces
$\sum_i u_i^2<\infty$, so the $u_i$ can be normalized and used as
Schmidt coefficients in the PV expression.  Lemma~\ref{lem:no-equality-chain}
then shows that the remaining optimality argument instead forces them
to grow geometrically.  This incompatibility is the contradiction.

We shall use the coarse bounds
\begin{equation}
 \frac14<\betaPV<\frac13,
\label{eq:coarse-bounds}
\end{equation}
proved in Appendix~\ref{app:technical-cases}.  The lower bound guarantees
that a matrix satisfying $\Phi(\theta)=\betaPV$ has a positive
off-diagonal entry;
the upper bound is used only for the exceptional endpoint $c=\pm1$.

\begin{lemma}[Consequences of equality for a finite spectral-weight matrix]
\label{lem:equality-chain}
Suppose a finite spectral-weight matrix satisfies
$\Phi(\theta)=\betaPV$.  Then there is a sequence of spectral labels
$(c_i)$, drawn
from the matrix's row and column index set, and positive auxiliary
coefficients $(u_i)$, indexed by a set $I\subseteq\mathbb Z$ that is
finite, one-sided, or all of $\mathbb Z$.  The sequence is obtained by
selecting positive entries $\theta_{c_{i-1}c_i}>0$; this selection need not
be unique.  The auxiliary coefficients $u_i$, defined from
the row and column sums of the whole matrix, satisfy
\begin{equation}
 0<\sum_{i\in I}u_i^2<\infty,
\label{eq:equality-l2}
\end{equation}
and satisfy the recurrence
\begin{equation}
 \betaPV u_i
 =d(c_{i-1},c_i)u_i
 +\frac{s(c_{i-1})}{2}u_{i-1}
 +\frac{s(c_i)}{2}u_{i+1}
\label{eq:pv-equality-recurrence}
\end{equation}
at every position (a coefficient beyond a finite endpoint is set
to zero), and any infinite end has an eventually constant label in
$(-1,1)$ and an eventually constant coefficient ratio; on an infinite
left end,
\begin{equation}
 c_i=c_-,
 \qquad
 \frac{u_{i+1}}{u_i}=\kappa>1
\label{eq:eventual-left-constancy}
\end{equation}
for all sufficiently negative $i$.  On an infinite right end there are
$c_+\in(-1,1)$ and $0<\kappa_+<1$ such that
\[
 c_i=c_+,
 \qquad
 \frac{u_{i+1}}{u_i}=\kappa_+
\]
for all sufficiently positive $i$.
\end{lemma}

\begin{proof}
Suppose a finite spectral-weight matrix $\theta$ satisfies
$\Phi(\theta)=\betaPV$.

Because this matrix attains the largest possible value of $\Phi$, its
one-sided variations give the first-order condition
Eq.~\eqref{eq:directional-bound}.  The first two steps both use this
condition.  Step~1 combines it with an exact weighted-average identity
to obtain equality at every positive entry.  Step~2 applies the same
condition to two crossed pairs to order the positive entries.

\paragraph*{Step 1: the first-order condition and equality at positive entries.}
For an ordered pair $(a,b)$ with $R_a>0$ and $C_b>0$, consider
\[
 \theta_t=(1-t)\theta+t\delta_{ab},
\]
where $\delta_{ab}$ has all its weight at $(a,b)$.  Define
\begin{equation}
 h(a,b):=d(a,b)
 +\frac{s(a)}2\sqrt{\frac{C_a}{R_a}}
 +\frac{s(b)}2\sqrt{\frac{R_b}{C_b}}.
\label{eq:directional-value}
\end{equation}
 It depends on the
row and column sums of the whole matrix and is introduced because
\begin{equation}
 \left.\frac{d}{dt}\Phi(\theta_t)\right|_{0^+}
 =h(a,b)-\betaPV\leq0,
 \qquad h(a,b)\leq\betaPV.
\label{eq:directional-bound}
\end{equation}
The inequality holds because every such matrix satisfies the bound in
Lemma~\ref{lem:finite-sequence-averaging}, while $\Phi(\theta)$
already equals $\betaPV$.  A term with zero numerator is interpreted
as zero.  This also gives the correct one-sided derivative when
$C_a=0$ or $R_b=0$, since the corresponding square-root
contribution remains zero.  (Such a zero marginal forces $a\neq b$
here because $R_a$ and $C_b$ are positive.)

The exact identity
\begin{align*}
 \sum_{\theta_{ab}>0}\theta_{ab}h(a,b)
 &=\sum_{a,b}d(a,b)\theta_{ab}
   +\frac12\sum_a s(a)\sqrt{R_aC_a}\\
 &\quad+\frac12\sum_b s(b)\sqrt{R_bC_b}\\
 &=\Phi(\theta)
\end{align*}
follows from the definitions of the row and column sums.  Since
$\sum_{a,b}\theta_{ab}=1$, it shows explicitly that $\Phi(\theta)$ is
the $\theta$-weighted average of the numbers $h(a,b)$.
Equation~\eqref{eq:directional-bound}
bounds each of these numbers by $\betaPV$, while their weighted
average equals $\betaPV$.  Hence every entry with positive weight
must satisfy
\begin{equation}
 \theta_{ab}>0
 \quad\Longrightarrow\quad
 h(a,b)=\betaPV.
\label{eq:positive-entry-equality}
\end{equation}

\paragraph*{Step 2: the same condition orders the positive entries.}
If $\theta_{ab}>0$, $\theta_{a'b'}>0$, and
$a<a'$, then $b\leq b'$.  Indeed, apply
Eq.~\eqref{eq:directional-bound} to the crossed pairs $(a,b')$ and
$(a',b)$, and subtract the two instances of
Eq.~\eqref{eq:positive-entry-equality}.  The square-root terms cancel,
leaving
\[
 (a-a')(b'-b)\leq0,
\]
and therefore $b\leq b'$.

Now suppose
\[
 \theta_{c_0c_1}>0,\qquad
 \theta_{c_1c_2}>0,\qquad\ldots,
\]
with consecutive $c_i$'s distinct.  If $c_0<c_1$, the ordering gives
$c_1<c_2$, then $c_2<c_3$, and so on.  If $c_0>c_1$, all the
inequalities are reversed.  Thus the $c_i$'s are strictly monotone, so
they never repeat and such a sequence must be finite.

At least one positive entry is off the diagonal.  Otherwise
$R_c=C_c=\theta_{cc}$, and
\[
 \Phi(\theta)
 =\sum_c\theta_{cc}[d(c,c)+s(c)]
 =\sum_c\theta_{cc}[s(c)-s(c)^2]
 \leq\frac14,
\]
contrary to Eq.~\eqref{eq:coarse-bounds}.

\paragraph*{Step 3: select matching positive entries.}
Start from a positive off-diagonal entry and then select positive
entries whose indices match:
\[
 \theta_{c_0c_1}>0,\qquad
 \theta_{c_1c_2}>0,\qquad\ldots,\qquad
 \theta_{c_{L-1}c_L}>0.
\]
Thus the second index of each selected entry is the first index of the
next one.  If several continuations are available, choose any one and
extend it maximally in both directions
through entries with distinct successive labels.  By Step~2
the distinct labels are strictly monotone.  Because the matrix has
only finitely many indices, the sequence $c_0,\ldots,c_L$ is finite.
It need not contain every positive entry and need not be unique.
Entries not selected remain part of $\theta$, and hence still affect
the row and column sums used below.

By maximality and the ordering proved in Step~2, the only possible
positive entry with second coordinate $c_0$ is $(c_0,c_0)$---any
other such entry would extend the sequence one step further left.  Hence
\[
 C_{c_0}=\theta_{c_0c_0}<R_{c_0},
\]
where strictness follows from $\theta_{c_0c_1}>0$.  Similarly, the
only possible positive entry with first coordinate $c_L$ is
$(c_L,c_L)$, so
\[
 R_{c_L}=\theta_{c_Lc_L}<C_{c_L}.
\]
If $\theta_{c_0c_0}>0$, extend the sequence by $c_0$ indefinitely
to the left; otherwise retain a finite endpoint.  Treat $c_L$
analogously on the right.  Every consecutive pair in the resulting
finite or infinite sequence is a positive entry of $\theta$.  Schematically,
\[
 \ldots,c_0,c_0,c_0,c_1,\ldots,c_{L-1},c_L,c_L,c_L,\ldots,
\]
where either constant end is absent when the sequence has a finite
endpoint.  An infinite end means only that the same positive diagonal
entry $(c_0,c_0)$, or $(c_L,c_L)$, is used repeatedly; it does not
turn the original finite-dimensional strategy into an
infinite-dimensional one.

\paragraph*{Step 4: define auxiliary coefficients from the full marginals.}
We now introduce the auxiliary numbers $u_i$.  Fix one of them to an
arbitrary positive value and define all the others recursively: whenever
positions $i$ and $i+1$ both belong to the sequence, set
\begin{equation}
 \frac{u_{i+1}}{u_i}
 =\sqrt{\frac{R_{c_i}}{C_{c_i}}}.
\label{eq:route-auxiliary-coefficients}
\end{equation}
Thus $u_i$ is associated with the selected entry
$(c_{i-1},c_i)$, as pictured by
\[
 c_0\xleftrightarrow{\ u_1\ }c_1
 \xleftrightarrow{\ u_2\ }\cdots
 \xleftrightarrow{\ u_L\ }c_L.
\]
For the selected entry $(c_{i-1},c_i)$, this definition gives
\[
 \sqrt{\frac{C_{c_{i-1}}}{R_{c_{i-1}}}}
 =\frac{u_{i-1}}{u_i},
 \qquad
 \sqrt{\frac{R_{c_i}}{C_{c_i}}}
 =\frac{u_{i+1}}{u_i}.
\]
Consequently Eq.~\eqref{eq:positive-entry-equality} becomes
\[
 \betaPV
 =d(c_{i-1},c_i)
 +\frac{s(c_{i-1})}{2}\frac{u_{i-1}}{u_i}
 +\frac{s(c_i)}{2}\frac{u_{i+1}}{u_i}.
\]
Multiplying by $u_i$ gives
Eq.~\eqref{eq:pv-equality-recurrence}.  Thus the $u_i$ satisfy exactly the
equation that the Schmidt coefficients of a PV strategy at value $\betaPV$
would have to satisfy.  On a constant left end the
ratio is
$\kappa=\sqrt{R_{c_0}/C_{c_0}}>1$, so the auxiliary coefficients decay
geometrically as $i\to-\infty$; on a constant right end the ratio is
smaller than one, so they decay as $i\to+\infty$.  Only finitely many
terms lie between the ends, proving Eq.~\eqref{eq:equality-l2}.  At a
finite left endpoint $C_{c_0}=0$, and at a finite right endpoint
$R_{c_L}=0$; the corresponding square-root term of $h$ vanishes,
consistent with setting the missing coefficient to zero in the
recurrence.

Finally, an eventually constant end label cannot equal $\pm1$:
Eq.~\eqref{eq:positive-entry-equality} for the corresponding diagonal
entry would give $h=0$, whereas $\betaPV>1/4$.  This establishes
Eq.~\eqref{eq:eventual-left-constancy} and its right-end analogue.
\end{proof}

We now evaluate the PV expression on the spectral labels and auxiliary
coefficients supplied by Lemma~\ref{lem:equality-chain}:
\begin{equation}
 \mathcal P_I(c,u):=
 \frac{\displaystyle
  \sum_{i\in I}d(c_{i-1},c_i)u_i^2
  +\sum_{\substack{i\in I\\i+1\in I}}
     s(c_i)u_i u_{i+1}}
 {\displaystyle\sum_{i\in I}u_i^2}.
\label{eq:infinite-pv-value}
\end{equation}
This step concerns only the auxiliary numbers; it does not rewrite the
original state and measurements as a PV strategy.  When $I$ is finite,
adding zero coefficients at the endpoints and normalizing the $u_i$
makes them the Schmidt coefficients of a genuine finite PV strategy,
whose value is the expression above.  When $I$ is infinite, we use only
this convergent numerical expression and its finite truncations; no
infinite-dimensional realization is assumed.  The sums converge absolutely
because $d$ is
bounded and
$2u_i u_{i+1}\leq u_i^2+u_{i+1}^2$.
Multiplying Eq.~\eqref{eq:pv-equality-recurrence} by $u_i$ and
summing over $i$ gives
\begin{equation}
 \betaPV\sum_{i\in I}u_i^2
 =\sum_{i\in I}d(c_{i-1},c_i)u_i^2
 +\sum_{\substack{i\in I\\i+1\in I}}s(c_i)u_i u_{i+1}.
\end{equation}
Indeed, each product $u_i u_{i+1}$ occurs twice, once with coefficient
$s(c_i)/2$ from each adjacent recurrence.  Hence
$\mathcal P_I(c,u)=\betaPV$.

\begin{lemma}[The equality conditions are inconsistent]
\label{lem:no-equality-chain}
No sequence has all the properties listed in
Lemma~\ref{lem:equality-chain}.
\end{lemma}

\begin{proof}
Assume for contradiction that such a sequence exists; then
$\mathcal P_I(c,u)=\betaPV$.

The contradiction has two parts.  First, the sequence cannot stop at
either end, because appending a sufficiently small auxiliary coefficient would
increase its value.  Second, varying each $c_i$ gives one additional
equation.  Together with Eq.~\eqref{eq:pv-equality-recurrence}, it forces
the constant left tail to continue throughout the sequence, in conflict
with the right end.

\paragraph*{Step 1: finite truncations are bounded by $\betaPV$.}
For every sequence with square-summable auxiliary coefficients,
$\mathcal P_I$ in Eq.~\eqref{eq:infinite-pv-value} is at most
$\betaPV$.  A finite sequence belongs, after zero padding, to the
variational family in Eq.~\eqref{eq:beta-pv}.  For an infinite sequence,
truncations to growing finite windows, after the same zero padding, are
genuine finite PV chains,
and their values converge to $\mathcal P_I(c,u)$ by the
absolute convergence noted below Eq.~\eqref{eq:infinite-pv-value}.
The same conclusion holds after changing finitely many labels or
auxiliary coefficients.

\paragraph*{Step 2: the sequence cannot stop.}
Suppose the sequence has an endpoint, say on the right, and put
$i_0=\max I$.  If $|c_{i_0}|<1$, set $c_{i_0+1}=-1$ and append a
small positive auxiliary coefficient $\varepsilon$.  The new nearest-neighbor term
$s(c_{i_0})u_{i_0}\varepsilon$ is positive and linear in
$\varepsilon$, while the new diagonal term and the normalization
change only quadratically.  The value therefore increases above
$\betaPV$, contradicting Step~1.  If $c_{i_0}=\pm1$, then
$s(c_{i_0})=0$; instead use the endpoint perturbation of
Appendix~\ref{app:technical-cases}, which moves $c_{i_0}$ into
$(-1,1)$ while appending an auxiliary coefficient.  Positivity of its linear
coefficient is guaranteed by $\betaPV<1/3$.  A left endpoint is
excluded by reversing the sequence and replacing every $c_i$ by
$-c_i$, which leaves Eq.~\eqref{eq:infinite-pv-value} unchanged.
Hence the sequence continues indefinitely in both directions.

\paragraph*{Step 3: what varying one spectral label requires.}
No internal label can equal $\sigma=\pm1$.  Replacing it by
$\sigma(1-\delta)$ creates a positive nearest-neighbor term of order
$\sqrt\delta$, while the affected diagonal terms change only by
$O(\delta)$.  For small $\delta>0$ this would raise the value above
$\betaPV$.  Thus $-1<c_i<1$ for every $i$.

Changing one $c_i$ affects only the two diagonal terms containing
that label and the nearest-neighbor term between their
auxiliary coefficients.  Both signs of a
sufficiently small change are allowed, and the current value is already
the supremum, so $c_i$ is an interior maximizer.  Differentiating the
numerator of Eq.~\eqref{eq:infinite-pv-value} gives
\begin{equation}
 \left(c_{i-1}-\frac12\right)u_i^2
 +\left(c_{i+1}+\frac12\right)u_{i+1}^2
 -\frac{c_i}{s(c_i)}u_i u_{i+1}=0.
\label{eq:parameter-stationarity}
\end{equation}

\paragraph*{Step 4: the constant left tail extends through the sequence.}
Set $\kappa_i=u_{i+1}/u_i>0$.  Equations
\eqref{eq:pv-equality-recurrence} and
\eqref{eq:parameter-stationarity} become
\begin{align}
 \betaPV
 &=d(c_{i-1},c_i)
 +\frac{s(c_{i-1})}{2\kappa_{i-1}}
 +\frac{s(c_i)}2\kappa_i,
\label{eq:ratio-recurrence}\\
 0
 &=c_{i-1}-\frac12
 +\left(c_{i+1}+\frac12\right)\kappa_i^2
 -\frac{c_i}{s(c_i)}\kappa_i.
\label{eq:label-recurrence}
\end{align}
Because $s(c_i)>0$, the first equation uniquely determines
$\kappa_i$ from $(c_{i-1},c_i,\kappa_{i-1})$; the second then
uniquely determines $c_{i+1}$.  By
Eq.~\eqref{eq:eventual-left-constancy}, there is a position $i_0$
with $(c_i,\kappa_i)=(c_-,\kappa)$ for all $i\leq i_0$, where
$\kappa>1$.  For sufficiently negative $i$, therefore,
$(c_{i-1},c_i,\kappa_{i-1})=(c_-,c_-,\kappa)$.  Substitution into the two recurrences gives the
same next values, $\kappa_i=\kappa$ and $c_{i+1}=c_-$.  Uniqueness
and induction give
\[
 c_i=c_-,
 \qquad
 \kappa_i=\kappa
 \quad\text{for every }i\geq i_0.
\]
Thus $u_{i+1}=\kappa u_i$ for all $i\geq i_0$.  Since
$\kappa>1$, this already contradicts the ratio
$\kappa_+<1$ required sufficiently far to the right by
Lemma~\ref{lem:equality-chain}.  Equivalently, the auxiliary coefficients grow
geometrically to the right and cannot satisfy
Eq.~\eqref{eq:equality-l2}.
\end{proof}

\begin{remark}
\label{rem:scope}
The contradiction is specific to the sequence constructed from a finite
spectral-weight matrix,
whose spectral label and coefficient ratio are eventually \emph{exactly} constant
at each infinite end.  It does not apply to arbitrary infinite sequences
whose spectral labels and coefficient ratios merely converge, and it makes no
claim about attainment by infinite-dimensional strategies.
\end{remark}

The proof above was written in elementary terms for quantum-information
readers.  Some of its steps nevertheless have standard interpretations in
other areas of mathematics.  In particular, the PV path construction of
Sec.~IV can be interpreted in terms of Markov chains.  In Sec.~V, the
three-term recurrence is the eigenvalue equation for a Jacobi operator, and
the identity used in Step~2 of Lemma~4 is called an anti-Monge relation in
optimization theory.  These connections are not needed for the proof itself,
but provide useful mathematical context; see
Refs.~\cite{NorrisMarkovChains,TeschlJacobi,BurkardKlinzRudolf}.

\begin{proof}[Proof of Theorem~\ref{thm:nonattainment}]
If a finite-dimensional strategy attained $\Istar$, its spectral-weight
matrix (Lemma~\ref{lem:table-bound}) would satisfy
\[
 \Istar=I_{3322}(\mathsf S)
 \leq\Phi(\theta)
 \leq\betaPV=\Istar.
\]
Thus $\Phi(\theta)=\betaPV$.  Lemma~\ref{lem:equality-chain} would then
produce the sequence of spectral labels and auxiliary coefficients described above,
which Lemma~\ref{lem:no-equality-chain} excludes.
\end{proof}

Corollary~\ref{cor:nonclosure} follows by a compactness argument.
Since $\Istar$ is a supremum, there are finite-dimensional
strategies whose values converge to it; by
Theorem~\ref{thm:variational}, finite PV chains already supply such a
sequence.  Each strategy determines a correlation, the list of
probabilities $P(\alpha\beta|xy)$, and these lists lie in a compact set:
finitely many probabilities, each in $[0,1]$.  Some subsequence of
the correlations therefore converges to a limit point $P_\infty$,
which lies in $\mathcal C_{qa}=\overline{\mathcal C_q}$ by the
definition of the closure.  The Bell value is linear in the
correlation, hence continuous, so $I_{3322}(P_\infty)=\Istar$.  If
$P_\infty$ belonged to $\mathcal C_q$, the finite-dimensional
strategy producing it would attain the supremum, which
Theorem~\ref{thm:nonattainment} excludes.  Hence
$P_\infty\in\mathcal C_{qa}\setminus\mathcal C_q$, which proves
Eq.~\eqref{eq:nonclosure-main}.

For Corollary~\ref{cor:dimension-witness}, fix $k$ and consider the
strategies whose local dimensions are both at most $k$.  After padding
the smaller space, these form a compact set in one fixed dimension, on
which the Bell value is continuous; hence the restricted maximum
$\beta_k$ is attained.  Theorem~\ref{thm:nonattainment} then gives
$\beta_k<\Istar$, proving Eq.~\eqref{eq:dimension-gap}.

\section{Discussion}
\label{sec:consequences}

The results establish nonclosure at a familiar facet of the local polytope
in the smallest binary-outcome scenario permitted by Beigi's two-setting
closure result~\cite{Beigi}.  The $I_{3322}$ attainment question and the
general connection between closedness and attainment were highlighted in
Ref.~\cite{DeltaGame}.
Concretely, $P_\infty$ can be approximated arbitrarily well by
finite-dimensional quantum correlations but cannot itself be generated by a
finite-dimensional strategy.

The result also makes $I_{3322}$ an unbounded dimension witness
\cite{DimWitness,VertesiPalDimension,NV}: a
certified value above $\beta_k$ proves that the local dimension
exceeds $k$ on at least one side.  Proposition~\ref{prop:qubit-bound}
proves $\beta_2=1/4$, so a certified value above $1/4$ certifies a
local dimension of at least three on one side.  Exact thresholds for
$k\geq3$, and their asymptotic approach to $\Istar$, remain open.

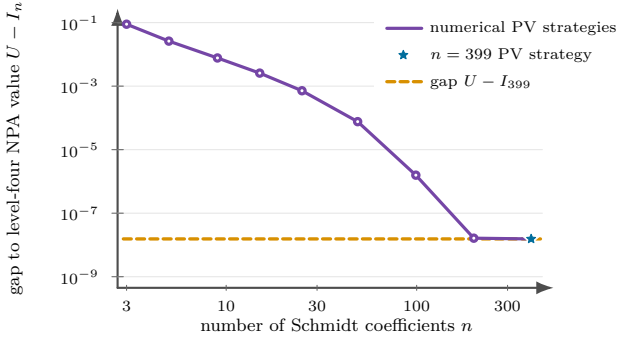
\begin{figure}[t]
\centering
\begin{tikzpicture}[v3 every figure, scale=0.84, transform shape]
  \begin{scope}[shift={(-1.28cm,-0.48cm)}, x=3.0cm, y=0.50cm]
    \draw[->, thick, black!70] (0.43,0.65) -- (2.72,0.65);
    \draw[->, thick, black!70] (0.43,0.65) -- (0.43,9.45);

    \foreach \xx/\lab in {
      0.4771/3,
      1.0000/10,
      1.4771/30,
      2.0000/100,
      2.4771/300
    } {
      \draw[black!60] (\xx,0.57) -- (\xx,0.73);
      \node[font=\scriptsize, anchor=north] at (\xx,0.48) {$\lab$};
    }
    \foreach \yy/\lab in {
      1/{10^{-9}},
      3/{10^{-7}},
      5/{10^{-5}},
      7/{10^{-3}},
      9/{10^{-1}}
    } {
      \draw[black!10] (0.43,\yy) -- (2.66,\yy);
      \draw[black!60] (0.40,\yy) -- (0.46,\yy);
      \node[font=\scriptsize, anchor=east] at (0.37,\yy) {$\lab$};
    }

    \draw[vthreerelaxed, densely dashed, line width=1.25pt]
      (0.46,2.18994) -- (2.65,2.18994);

    \draw[vthreepath, very thick]
      plot coordinates {
        (0.4771,8.94959)
        (0.6990,8.41724)
        (0.9542,7.88671)
        (1.1761,7.41007)
        (1.3979,6.85484)
        (1.6902,5.88191)
        (1.9956,4.19676)
        (2.2999,2.21291)
        (2.6009,2.18994)
      };
    \foreach \xx/\yy in {
      0.4771/8.94959,
      0.6990/8.41724,
      0.9542/7.88671,
      1.1761/7.41007,
      1.3979/6.85484,
      1.6902/5.88191,
      1.9956/4.19676,
      2.2999/2.21291
    } {
      \filldraw[fill=white, draw=vthreepath, very thick]
        (\xx,\yy) circle[radius=1.4pt];
    }
    \node[star, star points=5, star point ratio=2.2, fill=white,
          draw=white, minimum size=8pt, inner sep=0pt]
      at (2.6009,2.18994) {};
    \node[star, star points=5, star point ratio=2.2, fill=vthreeexact,
          draw=vthreeexact, minimum size=4.2pt, inner sep=0pt]
      at (2.6009,2.18994) {};

    \node[font=\footnotesize, anchor=north] at (1.58,-0.10)
      {number of Schmidt coefficients $n$};
    \node[font=\footnotesize, rotate=90, anchor=south] at (-0.03,5.05)
      {gap to level-four NPA value $U-I_n$};
  \end{scope}

  \fill[white] (4.05,2.82) rectangle (7.22,4.14);
  \draw[vthreepath, very thick] (4.25,3.93) -- ++(0.55,0)
    node[right,black,font=\scriptsize] {numerical PV strategies};
  \node[star, star points=5, star point ratio=2.2, fill=vthreeexact,
        draw=vthreeexact, minimum size=4.2pt, inner sep=0pt]
    at (4.51,3.51) {};
  \node[anchor=west,font=\scriptsize] at (4.82,3.51)
    {$n=399$ PV strategy};
  \draw[vthreerelaxed, densely dashed, line width=1.25pt]
    (4.25,3.09) -- ++(0.55,0)
    node[right,black,font=\scriptsize] {gap $U-I_{399}$};
\end{tikzpicture}
 \caption{\label{fig:convergence}
Numerical convergence of finite PV strategies.  We plot $U-I_n$,
where $I_n$ is the value of a PV strategy with $n$ Schmidt coefficients.
All points are numerical optimizations, the star marks $n=399$, and the
dashed line marks $U-I_{399}$.  Here $U=0.25087540$ is the ordinary
level-four NPA value reported in Ref.~\cite{AraujoEtAl2026}.}
\end{figure}

The main open question concerns infinite dimensions.  In the standard
hierarchy
$\mathcal C_q\subseteq\mathcal C_{qs}\subseteq\mathcal C_{qa}
\subseteq\mathcal C_{qc}$
\cite{Tsirelson93,ScholzWerner,Fritz,JNPPSW}, $\mathcal C_{qs}$
collects the correlations of possibly infinite-dimensional tensor-product
strategies and $\mathcal C_{qc}$ those of commuting-operator
strategies.  Equation~\eqref{eq:nonclosure-main} forces at least one
of the separations $\mathcal C_q\neq\mathcal C_{qs}$ and
$\mathcal C_{qs}\neq\mathcal C_{qa}$ already at $(3,3,2,2)$.
The compactness argument above produces a maximizing correlation in
$\mathcal C_{qa}$; whether \emph{any} maximizing correlation lies in
$\mathcal C_{qs}$ is the attainment question our result leaves open.
If one does, the supremum is
attained on an infinite-dimensional tensor product, and maximal
violation of $I_{3322}$ requires infinite-dimensional entanglement,
completing the picture conjectured by Coladangelo and Stark
\cite{ColadangeloStark}.  If no maximizing correlation lies in
$\mathcal C_{qs}$, then $\mathcal C_{qs}\neq\mathcal C_{qa}$, and
no tensor-product strategy of any dimension maximally violates
$I_{3322}$.  To our knowledge, whether the commuting-operator value of
$I_{3322}$ exceeds $\Istar$ is a further open question.

Several further questions remain.  First, how fast does the
dimension-restricted gap $\Istar-\beta_k$ close as $k$ grows?
The decay rate determines the dimension needed to approach the
supremum to a given precision.  To illustrate the convergence with dimension,
Fig.~\ref{fig:convergence} shows our fixed-length PV optimizations up to
$n=399$.  These values come from direct numerical optimization of
Eq.~\eqref{eq:pv-value}.  For comparison, the best NPA value reported in
Ref.~\cite{AraujoEtAl2026} is $0.25087538$.  Second, does the supremum admit a
simpler analytic description?  Its value is known numerically to high
precision, and we are not aware of a closed form
\cite{AraujoMathOverflow}.  Third, is there a
robust rigidity statement for nearly maximal finite-dimensional strategies
\cite{SupicBowles}---that is, must every strategy with value close to
$\Istar$ be close, in a suitable sense, to some finite PV chain?
Theorem~\ref{thm:variational} is a statement about values; a
self-testing statement would identify the strategies behind them.
Finally, the reduction of Secs.~\ref{sec:strategy-table}
and~\ref{sec:table-to-pv} used the swap symmetry of the functional,
the linear elimination of the third measurements, and Jordan's lemma;
which other Bell functionals admit a comparable reduction to an
explicit variational family is open.
 
\begin{acknowledgments}
I thank Nicolas Gisin for helpful discussions.

\textbf{AI statement.}
I should explain how this paper arrived.  On 25 July 2026 I gave the
finite-dimensional non-attainment conjecture to OpenAI Codex
(GPT-5.6-sol).  It found a two-qubit bound and then politely stopped.  I
replied, ``I believe you gave up too quickly and you will succeed.  Make a very
strong push to either prove or disprove the conjecture.''  Roughly three hours
later, the first end-to-end proof appeared.  Perhaps the pep talk mattered:
Anthropic recently reported that variants of ``keep going'' and ``believe in
yourself'' seemed to help Claude overcome its initial skepticism during a
Riemann-hypothesis experiment~\cite{AnthropicRiemann}.  The proof was not yet
written in anything a quantum-information reader would recognize as English:
it spoke of a ``classical shadow,'' a ``symmetric transport plan,'' a
``positive cocycle,'' ``heteroclinic orbits,'' ``spectral edges,'' and a
``domain wall.''  Much of the work since---with Anthropic's Claude also serving
as translator and hostile referee---has been to turn that first draft into the
argument above.  The same tools were then used to formalize the argument in Lean~4 (Appendix~\ref{app:lean}).

During the same period, Seth Douglas---who states on his GitHub profile that
he is building a ``human+AI automated research engine'' in his spare
time---posted a vast
AI-assisted repository attacking the same
problem~\cite{DouglasI3322}.  The repository seems to take a significantly
different approach from ours, but I have not been able to verify its
correctness; nothing in the present paper depends on it.  I nevertheless
invite readers to look at it.  The experience of trying to read the
repository---together with the equally bewildering state of our first
draft---illustrates the continuing human task of digesting, reconstructing,
checking, and explaining machine-generated proofs, or ``proof
digestion''~\cite{TaoMathematicsAI}.  I set the problem and standards of
evidence and take responsibility for the final claims.  These are unusual
circumstances, but I feel an obligation to share the result with the community.
\end{acknowledgments}

\bibliographystyle{apsrev4-2}
\bibliography{references}

@article{Bell,
  author  = {Bell, John S.},
  title   = {On the {Einstein Podolsky Rosen} paradox},
  journal = {Physics Physique Fizika},
  volume  = {1},
  number  = {3},
  pages   = {195--200},
  year    = {1964},
  doi     = {10.1103/PhysicsPhysiqueFizika.1.195}
}

@article{CHSH,
  author  = {Clauser, John F. and Horne, Michael A. and Shimony, Abner
             and Holt, Richard A.},
  title   = {Proposed experiment to test local hidden-variable theories},
  journal = {Physical Review Letters},
  volume  = {23},
  number  = {15},
  pages   = {880--884},
  year    = {1969},
  doi     = {10.1103/PhysRevLett.23.880}
}

@article{ReviewBell,
  author  = {Brunner, Nicolas and Cavalcanti, Daniel and Pironio, Stefano
             and Scarani, Valerio and Wehner, Stephanie},
  title   = {Bell nonlocality},
  journal = {Reviews of Modern Physics},
  volume  = {86},
  number  = {2},
  pages   = {419--478},
  year    = {2014},
  doi     = {10.1103/RevModPhys.86.419},
  eprint  = {1303.2849},
  archivePrefix = {arXiv},
  primaryClass  = {quant-ph}
}

@article{Slofstra,
  author  = {Slofstra, William},
  title   = {The set of quantum correlations is not closed},
  journal = {Forum of Mathematics, Pi},
  volume  = {7},
  pages   = {e1},
  year    = {2019},
  doi     = {10.1017/fmp.2018.3},
  eprint  = {1703.08618},
  archivePrefix = {arXiv},
  primaryClass  = {quant-ph}
}

@article{SlofstraTsirelson,
  author  = {Slofstra, William},
  title   = {Tsirelson's problem and an embedding theorem for groups arising
             from non-local games},
  journal = {Journal of the American Mathematical Society},
  volume  = {33},
  number  = {1},
  pages   = {1--56},
  year    = {2020},
  doi     = {10.1090/jams/929},
  eprint  = {1606.03140},
  archivePrefix = {arXiv},
  primaryClass  = {quant-ph}
}

@article{DPP,
  author  = {Dykema, Ken and Paulsen, Vern I. and Prakash, Jitendra},
  title   = {Non-closure of the set of quantum correlations via graphs},
  journal = {Communications in Mathematical Physics},
  volume  = {365},
  number  = {3},
  pages   = {1125--1142},
  year    = {2019},
  doi     = {10.1007/s00220-019-03301-1},
  eprint  = {1709.05032},
  archivePrefix = {arXiv},
  primaryClass  = {math.OA}
}

@article{ColadangeloStark,
  author  = {Coladangelo, Andrea and Stark, Jalex},
  title   = {An inherently infinite-dimensional quantum correlation},
  journal = {Nature Communications},
  volume  = {11},
  pages   = {3335},
  year    = {2020},
  doi     = {10.1038/s41467-020-17077-9},
  eprint  = {1804.05116},
  archivePrefix = {arXiv},
  primaryClass  = {quant-ph}
}

@article{Tsirelson93,
  author  = {Tsirelson, Boris S.},
  title   = {Some results and problems on quantum {Bell}-type inequalities},
  journal = {Hadronic Journal Supplement},
  volume  = {8},
  number  = {4},
  pages   = {329--345},
  year    = {1993}
}

@misc{MIPstarRE,
  author  = {Ji, Zhengfeng and Natarajan, Anand and Vidick, Thomas
             and Wright, John and Yuen, Henry},
  title   = {{MIP}$^*=${RE}},
  year    = {2020},
  eprint  = {2001.04383},
  archivePrefix = {arXiv},
  primaryClass  = {quant-ph},
  note    = {Abridged version: Commun. ACM \textbf{64}(11), 131--138
             (2021), doi:10.1145/3485628}
}

@article{Froissart,
  author  = {Froissart, Marcel},
  title   = {Constructive generalization of {Bell}'s inequalities},
  journal = {Il Nuovo Cimento B},
  volume  = {64},
  number  = {2},
  pages   = {241--251},
  year    = {1981},
  doi     = {10.1007/BF02903286}
}

@article{Sliwa2003,
  author  = {{\'S}liwa, Cezary},
  title   = {Symmetries of the {Bell} correlation inequalities},
  journal = {Physics Letters A},
  volume  = {317},
  number  = {3--4},
  pages   = {165--168},
  year    = {2003},
  doi     = {10.1016/S0375-9601(03)01115-0},
  eprint  = {quant-ph/0305190},
  archivePrefix = {arXiv}
}

@article{BrunnerGisin2008,
  author  = {Brunner, Nicolas and Gisin, Nicolas},
  title   = {Partial list of bipartite {Bell} inequalities with four
             binary settings},
  journal = {Physics Letters A},
  volume  = {372},
  number  = {18},
  pages   = {3162--3167},
  year    = {2008},
  doi     = {10.1016/j.physleta.2008.01.052},
  eprint  = {0711.3362},
  archivePrefix = {arXiv},
  primaryClass  = {quant-ph}
}

@article{CollinsGisin,
  author  = {Collins, Daniel and Gisin, Nicolas},
  title   = {A relevant two qubit {Bell} inequality inequivalent to the
             {CHSH} inequality},
  journal = {Journal of Physics A: Mathematical and General},
  volume  = {37},
  number  = {5},
  pages   = {1775--1787},
  year    = {2004},
  doi     = {10.1088/0305-4470/37/5/021},
  eprint  = {quant-ph/0306129},
  archivePrefix = {arXiv}
}

@article{PV,
  author  = {P{\'a}l, K{\'a}roly F. and V{\'e}rtesi, Tam{\'a}s},
  title   = {Maximal violation of a bipartite three-setting, two-outcome
             {Bell} inequality using infinite-dimensional quantum systems},
  journal = {Physical Review A},
  volume  = {82},
  number  = {2},
  pages   = {022116},
  year    = {2010},
  doi     = {10.1103/PhysRevA.82.022116},
  eprint  = {1006.3032},
  archivePrefix = {arXiv},
  primaryClass  = {quant-ph}
}

@article{NPA1,
  author  = {Navascu{\'e}s, Miguel and Pironio, Stefano and Ac{\'i}n, Antonio},
  title   = {Bounding the set of quantum correlations},
  journal = {Physical Review Letters},
  volume  = {98},
  number  = {1},
  pages   = {010401},
  year    = {2007},
  doi     = {10.1103/PhysRevLett.98.010401},
  eprint  = {quant-ph/0607119},
  archivePrefix = {arXiv}
}

@article{NPA2,
  author  = {Navascu{\'e}s, Miguel and Pironio, Stefano and Ac{\'i}n, Antonio},
  title   = {A convergent hierarchy of semidefinite programs characterizing
             the set of quantum correlations},
  journal = {New Journal of Physics},
  volume  = {10},
  number  = {7},
  pages   = {073013},
  year    = {2008},
  doi     = {10.1088/1367-2630/10/7/073013},
  eprint  = {0803.4290},
  archivePrefix = {arXiv},
  primaryClass  = {quant-ph}
}

@article{AraujoEtAl2026,
  author  = {Ara{\'u}jo, Mateus and Klep, Igor and Garner, Andrew J. P.
             and V{\'e}rtesi, Tam{\'a}s and Navascu{\'e}s, Miguel},
  title   = {First-order optimality conditions for non-commutative
             optimization problems},
  journal = {Foundations of Computational Mathematics},
  year    = {2026},
  month   = jul,
  doi     = {10.1007/s10208-026-09761-x},
  eprint  = {2311.18707},
  archivePrefix = {arXiv},
  primaryClass  = {quant-ph},
  note    = {Published online 7 July 2026}
}

@misc{AraujoMathOverflow,
  author       = {Ara{\'u}jo, Mateus},
  title        = {Analytic expression for the {Tsirelson} bound of the
                  {$I_{3322}$} inequality?},
  howpublished = {MathOverflow},
  year         = {2016},
  month        = jul,
  url          = {https://mathoverflow.net/questions/243624/analytic-expression-for-the-tsirelson-bound-of-the-i3322-inequality},
  note         = {Question posted 4 July 2016}
}

@article{DimWitness,
  author  = {Brunner, Nicolas and Pironio, Stefano and Ac{\'i}n, Antonio
             and Gisin, Nicolas and M{\'e}thot, Andr{\'e} Allan
             and Scarani, Valerio},
  title   = {Testing the dimension of {Hilbert} spaces},
  journal = {Physical Review Letters},
  volume  = {100},
  number  = {21},
  pages   = {210503},
  year    = {2008},
  doi     = {10.1103/PhysRevLett.100.210503},
  eprint  = {0802.0760},
  archivePrefix = {arXiv},
  primaryClass  = {quant-ph}
}

@article{NV,
  author  = {Navascu{\'e}s, Miguel and V{\'e}rtesi, Tam{\'a}s},
  title   = {Bounding the set of finite-dimensional quantum correlations},
  journal = {Physical Review Letters},
  volume  = {115},
  number  = {2},
  pages   = {020501},
  year    = {2015},
  doi     = {10.1103/PhysRevLett.115.020501},
  eprint  = {1412.0924},
  archivePrefix = {arXiv},
  primaryClass  = {quant-ph}
}

@article{MoroderEtAl,
  author  = {Moroder, Tobias and Bancal, Jean-Daniel
             and Liang, Yeong-Cherng and Hofmann, Martin
             and G{\"u}hne, Otfried},
  title   = {Device-independent entanglement quantification
             and related applications},
  journal = {Physical Review Letters},
  volume  = {111},
  number  = {3},
  pages   = {030501},
  year    = {2013},
  doi     = {10.1103/PhysRevLett.111.030501},
  eprint  = {1302.1336},
  archivePrefix = {arXiv},
  primaryClass  = {quant-ph}
}

@article{NavascuesDeLaTorreVertesi,
  author  = {Navascu{\'e}s, Miguel and de la Torre, Gonzalo
             and V{\'e}rtesi, Tam{\'a}s},
  title   = {Characterization of quantum correlations with local dimension
             constraints and its device-independent applications},
  journal = {Physical Review X},
  volume  = {4},
  number  = {1},
  pages   = {011011},
  year    = {2014},
  doi     = {10.1103/PhysRevX.4.011011},
  eprint  = {1308.3410},
  archivePrefix = {arXiv},
  primaryClass  = {quant-ph}
}

@article{VidickWehner,
  author  = {Vidick, Thomas and Wehner, Stephanie},
  title   = {More nonlocality with less entanglement},
  journal = {Physical Review A},
  volume  = {83},
  number  = {5},
  pages   = {052310},
  year    = {2011},
  doi     = {10.1103/PhysRevA.83.052310},
  eprint  = {1011.5206},
  archivePrefix = {arXiv},
  primaryClass  = {quant-ph}
}

@article{GigenaKaniewski,
  author  = {Gigena, Nicol{\'a}s and Kaniewski, J{\k e}drzej},
  title   = {Quantum value for a family of {$I_{3322}$}-like {Bell}
             functionals},
  journal = {Physical Review A},
  volume  = {106},
  number  = {1},
  pages   = {012401},
  year    = {2022},
  doi     = {10.1103/PhysRevA.106.012401},
  eprint  = {2203.01837},
  archivePrefix = {arXiv},
  primaryClass  = {quant-ph}
}

@article{BrunnerEtAlDetection,
  author  = {Brunner, Nicolas and Gisin, Nicolas and Scarani, Valerio
             and Simon, Christoph},
  title   = {Detection loophole in asymmetric {Bell} experiments},
  journal = {Physical Review Letters},
  volume  = {98},
  number  = {22},
  pages   = {220403},
  year    = {2007},
  doi     = {10.1103/PhysRevLett.98.220403},
  eprint  = {quant-ph/0702130},
  archivePrefix = {arXiv}
}

@article{PomaricoEtAl,
  author  = {Pomarico, Enrico and Bancal, Jean-Daniel
             and Sanguinetti, Bruno and Rochdi, Anas and Gisin, Nicolas},
  title   = {Various quantum nonlocality tests with a commercial
             two-photon entanglement source},
  journal = {Physical Review A},
  volume  = {83},
  number  = {5},
  pages   = {052104},
  year    = {2011},
  doi     = {10.1103/PhysRevA.83.052104},
  eprint  = {1101.2313},
  archivePrefix = {arXiv},
  primaryClass  = {quant-ph}
}

@article{SuDIQKD,
  author  = {Su, Hong-Yi},
  title   = {Monte Carlo approach to the evaluation of the security of
             device-independent quantum key distribution},
  journal = {New Journal of Physics},
  volume  = {25},
  number  = {12},
  pages   = {123036},
  year    = {2023},
  doi     = {10.1088/1367-2630/ad141a},
  eprint  = {2308.03030},
  archivePrefix = {arXiv},
  primaryClass  = {quant-ph}
}

@article{HsuEtAl,
  author  = {Hsu, Hsin-Yu and Tabia, Gelo Noel M. and Chen, Kai-Siang
             and Liu, Mu-En and V{\'e}rtesi, Tam{\'a}s and Brunner, Nicolas
             and Liang, Yeong-Cherng},
  title   = {Trading symmetry for {Hilbert}-space dimension in
             {Bell}-inequality violation},
  journal = {Quantum Science and Technology},
  volume  = {11},
  number  = {3},
  pages   = {035003},
  year    = {2026},
  doi     = {10.1088/2058-9565/ae6bb1},
  eprint  = {2601.02893},
  archivePrefix = {arXiv},
  primaryClass  = {quant-ph}
}

@article{Beigi,
  author  = {Beigi, Salman},
  title   = {Separation of quantum, spatial quantum, and approximate
             quantum correlations},
  journal = {Quantum},
  volume  = {5},
  pages   = {389},
  year    = {2021},
  doi     = {10.22331/q-2021-01-28-389},
  eprint  = {2004.11103},
  archivePrefix = {arXiv},
  primaryClass  = {quant-ph}
}

@article{LiangDoherty,
  author  = {Liang, Yeong-Cherng and Doherty, Andrew C.},
  title   = {Bounds on quantum correlations in {Bell}-inequality experiments},
  journal = {Physical Review A},
  volume  = {75},
  number  = {4},
  pages   = {042103},
  year    = {2007},
  doi     = {10.1103/PhysRevA.75.042103},
  eprint  = {quant-ph/0608128},
  archivePrefix = {arXiv}
}

@article{VertesiPalDimension,
  author  = {V{\'e}rtesi, Tam{\'a}s and P{\'a}l, K{\'a}roly F.},
  title   = {Bounding the dimension of bipartite quantum systems},
  journal = {Physical Review A},
  volume  = {79},
  number  = {4},
  pages   = {042106},
  year    = {2009},
  doi     = {10.1103/PhysRevA.79.042106},
  eprint  = {0812.1572},
  archivePrefix = {arXiv},
  primaryClass  = {quant-ph}
}

@article{PalVertesi2008,
  author  = {P{\'a}l, K{\'a}roly F. and V{\'e}rtesi, Tam{\'a}s},
  title   = {Efficiency of higher-dimensional {Hilbert} spaces for the
             violation of {Bell} inequalities},
  journal = {Physical Review A},
  volume  = {77},
  number  = {4},
  pages   = {042105},
  year    = {2008},
  doi     = {10.1103/PhysRevA.77.042105},
  eprint  = {0712.4320},
  archivePrefix = {arXiv},
  primaryClass  = {quant-ph}
}

@misc{ScholzWerner,
  author        = {Scholz, Volkher B. and Werner, Reinhard F.},
  title         = {Tsirelson's problem},
  year          = {2008},
  eprint        = {0812.4305},
  archivePrefix = {arXiv},
  primaryClass  = {math-ph}
}

@article{Fritz,
  author  = {Fritz, Tobias},
  title   = {Tsirelson's problem and {Kirchberg}'s conjecture},
  journal = {Reviews in Mathematical Physics},
  volume  = {24},
  number  = {5},
  pages   = {1250012},
  year    = {2012},
  doi     = {10.1142/S0129055X12500122},
  eprint  = {1008.1168},
  archivePrefix = {arXiv},
  primaryClass  = {math-ph}
}

@article{JNPPSW,
  author  = {Junge, Marius and Navascu{\'e}s, Miguel and Palazuelos, Carlos
             and P{\'e}rez-Garc{\'i}a, David and Scholz, Volkher B.
             and Werner, Reinhard F.},
  title   = {Connes' embedding problem and {Tsirelson}'s problem},
  journal = {Journal of Mathematical Physics},
  volume  = {52},
  number  = {1},
  pages   = {012102},
  year    = {2011},
  doi     = {10.1063/1.3514538},
  eprint  = {1008.1142},
  archivePrefix = {arXiv},
  primaryClass  = {math-ph}
}

@article{Jordan1875,
  author  = {Jordan, Camille},
  title   = {Essai sur la g{\'e}om{\'e}trie {\`a} {$n$} dimensions},
  journal = {Bulletin de la Soci{\'e}t{\'e} Math{\'e}matique de France},
  volume  = {3},
  pages   = {103--174},
  year    = {1875},
  doi     = {10.24033/bsmf.90}
}

@article{PironioEtAl2009,
  author  = {Pironio, Stefano and Ac{\'i}n, Antonio and Brunner, Nicolas
             and Gisin, Nicolas and Massar, Serge and Scarani, Valerio},
  title   = {Device-independent quantum key distribution secure against
             collective attacks},
  journal = {New Journal of Physics},
  volume  = {11},
  number  = {4},
  pages   = {045021},
  year    = {2009},
  doi     = {10.1088/1367-2630/11/4/045021},
  eprint  = {0903.4460},
  archivePrefix = {arXiv},
  primaryClass  = {quant-ph}
}

@article{BernardsGuhne,
  author  = {Bernards, Fabian and G{\"u}hne, Otfried},
  title   = {Generalizing optimal {Bell} inequalities},
  journal = {Physical Review Letters},
  volume  = {125},
  number  = {20},
  pages   = {200401},
  year    = {2020},
  doi     = {10.1103/PhysRevLett.125.200401},
  eprint  = {2005.08687},
  archivePrefix = {arXiv},
  primaryClass  = {quant-ph}
}

@article{PironioLifting,
  author  = {Pironio, Stefano},
  title   = {Lifting {Bell} inequalities},
  journal = {Journal of Mathematical Physics},
  volume  = {46},
  number  = {6},
  pages   = {062112},
  year    = {2005},
  doi     = {10.1063/1.1928727},
  eprint  = {quant-ph/0503179},
  archivePrefix = {arXiv}
}

@article{SupicBowles,
  author  = {{\v S}upi{\'c}, Ivan and Bowles, Joseph},
  title   = {Self-testing of quantum systems: a review},
  journal = {Quantum},
  volume  = {4},
  pages   = {337},
  year    = {2020},
  doi     = {10.22331/q-2020-09-30-337},
  eprint  = {1904.10042},
  archivePrefix = {arXiv},
  primaryClass  = {quant-ph}
}

@misc{DouglasI3322,
  author       = {Douglas, Seth},
  title        = {Exact {$I_{3322}$} quantum wall},
  year         = {2026},
  howpublished = {Public GitHub repository},
  note         = {\url{https://github.com/Apsiape/i3322-exact-wall};
                  author profile: \url{https://github.com/Apsiape}}
}

@article{DeltaGame,
  author  = {Dykema, Ken and Paulsen, Vern I. and Prakash, Jitendra},
  title   = {The {Delta} game},
  journal = {Quantum Information and Computation},
  volume  = {18},
  number  = {7\&8},
  pages   = {599--616},
  year    = {2018},
  doi     = {10.26421/QIC18.7-8-5},
  eprint  = {1707.06186},
  archivePrefix = {arXiv},
  primaryClass  = {math.OA}
}

@article{AltepeterEtAl,
  author  = {Altepeter, J. B. and Jeffrey, E. R. and Kwiat, P. G.
             and Tanzilli, S. and Gisin, N. and Ac{\'i}n, A.},
  title   = {Experimental methods for detecting entanglement},
  journal = {Physical Review Letters},
  volume  = {95},
  number  = {3},
  pages   = {033601},
  year    = {2005},
  doi     = {10.1103/PhysRevLett.95.033601}
}

@article{Mortimer2025,
  author  = {Mortimer, Luke},
  title   = {Bounding large-scale {Bell} inequalities},
  journal = {Physical Review A},
  volume  = {111},
  number  = {5},
  pages   = {052442},
  year    = {2025},
  doi     = {10.1103/PhysRevA.111.052442},
  eprint  = {2412.08532},
  archivePrefix = {arXiv},
  primaryClass  = {quant-ph}
}

@misc{FloraEtAl2026,
  author        = {Flora, Francesco and Matos, Losel and Heightman, Tim
                   and Kriv{\'a}chy, Tam{\'a}s and Garriga, Adan
                   and Ac{\'i}n, Antonio},
  title         = {Moment optimization in the
                   {Navascu{\'e}s--Pironio--Ac{\'i}n} hierarchy},
  year          = {2026},
  eprint        = {2607.14755},
  archivePrefix = {arXiv},
  primaryClass  = {quant-ph}
}

@misc{TaoMathematicsAI,
  author        = {Tao, Terence},
  title         = {Mathematics in the age of {AI}},
  year          = {2026},
  eprint        = {2608.16753},
  archivePrefix = {arXiv},
  primaryClass  = {math.HO}
}

@misc{AnthropicRiemann,
  author       = {{Anthropic}},
  title        = {Learning more about {Claude}'s mathematical capabilities},
  year         = {2026},
  month        = aug,
  howpublished = {Research note},
  note         = {\url{https://www.anthropic.com/research/riemann-zeta}}
}

@article{PitowskySvozil,
  author  = {Pitowsky, Itamar and Svozil, Karl},
  title   = {Optimal tests of quantum nonlocality},
  journal = {Physical Review A},
  volume  = {64},
  number  = {1},
  pages   = {014102},
  year    = {2001},
  doi     = {10.1103/PhysRevA.64.014102},
  eprint  = {quant-ph/0011060},
  archivePrefix = {arXiv}
}

@book{NorrisMarkovChains,
  author    = {J. R. Norris},
  title     = {Markov Chains},
  publisher = {Cambridge University Press},
  year      = {1997},
  doi       = {10.1017/CBO9780511810633}
}

@book{TeschlJacobi,
  author    = {Gerald Teschl},
  title     = {Jacobi Operators and Completely Integrable Nonlinear Lattices},
  series    = {Mathematical Surveys and Monographs},
  volume    = {72},
  publisher = {American Mathematical Society},
  address   = {Providence, RI},
  year      = {2000},
  isbn      = {978-0-8218-1940-1}
}

@article{BurkardKlinzRudolf,
  author  = {Rainer E. Burkard and Bettina Klinz and R{\"u}diger Rudolf},
  title   = {Perspectives of Monge Properties in Optimization},
  journal = {Discrete Applied Mathematics},
  volume  = {70},
  number  = {2},
  pages   = {95--161},
  year    = {1996},
  doi     = {10.1016/0166-218X(95)00103-X}
}

\appendix
\setcounter{secnumdepth}{2}
\section{The exact two-qubit threshold}
\label{app:qubit-bound}

\begin{proposition}
\label{prop:qubit-bound}
The maximum of $I_{3322}$ over two-qubit states and arbitrary binary
measurements is
\[
 \beta_2=\frac14.
\]
\end{proposition}

\begin{proof}
The Bell value is separately affine in every binary effect, whose
extreme points on a qubit are projections.  It is therefore enough to
consider projective measurements.  Put $F:=4(1+I_{3322})$, as in
Eq.~\eqref{eq:F-reflections}.

Suppose first that none of $a_1,a_2,b_1,b_2$ equals $\pm\1$, so
that on a qubit each has one $+1$ and one $-1$ eigenvector.
Their sums and differences can be written
\[
 \begin{gathered}
 p=2cP,\quad r=2\bar cR,\quad \bar c=\sqrt{1-c^2},\\
 q=2dQ,\quad \tau=2\bar dT,\quad \bar d=\sqrt{1-d^2}.
 \end{gathered}
\]
where $0\leq c,d\leq1$, and $P,R$ and $Q,T$ are pairs of
anticommuting $\pm1$-valued observables.  If a coefficient vanishes,
choose the unused observable arbitrarily.  For an arbitrary state, set
\[
 \begin{gathered}
 x=\langle P\rangle,\quad y=\langle Q\rangle,
 \quad z=\langle PQ\rangle,\\
 u=\langle Rb_3\rangle,\quad v=\langle a_3T\rangle.
 \end{gathered}
\]
Equation~\eqref{eq:F-reflections} becomes
\[
 \frac F2=cx-dy+2cdz+\bar c u+\bar d v.
\]
In each of the pairs $(P,Rb_3)$ and $(Q,a_3T)$, the two operators are
anticommuting $\pm1$-valued observables.  Hence their expectation
values lie in the unit disk:
\[
 x^2+u^2\leq1,
 \qquad
 y^2+v^2\leq1.
\]
Moreover, $P$ and $Q$ commute because they belong to different
parties.  Their joint $\pm1$ distribution gives
\[
 1-z=2\Pr(P\ne Q)\geq|x-y|.
\]
Consequently $F/2\leq G$, where
\[
 G=cx-dy+2cd(1-|x-y|)
   +\bar c\sqrt{1-x^2}+\bar d\sqrt{1-y^2}.
\]

Whenever the maximum occurs at $x=y=t$,
\[
 G\leq2cd+\sqrt{(c-d)^2+(\bar c+\bar d)^2}\leq\frac52.
\]
Indeed, write $c=\cos C$, $d=\cos D$, and
$w=\sin((C+D)/2)$.  The middle expression is at most
\[
 1+\cos(C+D)+2\sin\frac{C+D}{2}
 =2+2w-2w^2\leq\frac52.
\]

It remains only to see when the maximum can lie away from $x=y$.
For $x\leq y$, the terms depending on $x$ and $y$ are
\[
 c(1+2d)x+\bar c\sqrt{1-x^2},
 \qquad
 -d(1+2c)y+\bar d\sqrt{1-y^2}.
\]
Their unconstrained maximizers are respectively nonnegative and
nonpositive, so by concavity the maximum under $x\leq y$ lies on
$x=y$.  For
$x\geq y$, the corresponding linear coefficients are
$c(1-2d)$ and $d(2c-1)$.  If $c,d>1/2$, their unconstrained
maximizers violate the constraint $x\geq y$, so by concavity the maximum
again lies on $x=y$.
Otherwise, discard the constraint and maximize $x$ and $y$
independently to obtain
\[
 \begin{aligned}
 G\leq{}&2cd+\sqrt{1-4c^2d(1-d)}
              +\sqrt{1-4d^2c(1-c)}\\
 \leq{}&2+2cd(1-c-d+2cd)\leq\frac52.
 \end{aligned}
\]
Here the second line uses $\sqrt{1-t}\leq1-t/2$.  For the last
inequality, assume by symmetry that $c\leq1/2$; then
$1-c-d+2cd\leq1-c$, so the term added to $2$ is at most
$2c(1-c)\leq1/2$.  Thus $F\leq5$.

It remains to treat the case where one of $a_1,a_2,b_1,b_2$
equals $\pm\1$.  Let
\[
\begin{aligned}
 \widehat F={}&a_1+a_2-b_1-b_2+(a_1+a_2)(b_1+b_2)\\
 &+(a_2-a_1)b_3+a_3(b_2-b_1)
\end{aligned}
\]
be the Bell operator, so that $F=\langle\widehat F\rangle$.  The
functional is invariant under
\[
 a_1\leftrightarrow a_2,\quad b_3\mapsto-b_3,
\]
and under exchanging the parties while sending
$a_i\mapsto-b_i$ and $b_i\mapsto-a_i$.  It therefore suffices to
consider $a_1=\pm\1$.

If $a_1=\1$, then
\[
 \widehat F=
 \underbrace{a_2(b_1+b_2)+a_3(b_2-b_1)}_{\mathrm{CHSH}}
 +\underbrace{\1+a_2+(a_2-\1)b_3}_{\leq2\1},
\]
and hence $F\leq2\sqrt2+2<5$.  If $a_1=-\1$, equivalently
$A_1=0$, the original effect form becomes
\[
 \begin{aligned}
 \widehat I_{3322}={}&B_1(A_2-A_3-\1)
   +B_2(A_2+A_3-2\1)\\
 &+A_2(B_3-\1)\leq0.
 \end{aligned}
\]
Each bracket is negative semidefinite and commutes with the positive
effect multiplying it.  Thus these remaining cases also satisfy
$F<5$, and hence $I_{3322}\leq1/4$.

Equality is attained by
\[
 |\Phi^+\rangle=\frac{|00\rangle+|11\rangle}{\sqrt2}
\]
with observables $b_i=\beta_i\cdot\vec\sigma$ and
$a_i=\alpha_i\cdot\vec\sigma$, whose Bloch vectors are
\[
\begin{aligned}
 \beta_{1,2}&=(\sqrt3/2,\mp1/2,0),&
 \beta_3&=(0,0,1),\\
 \alpha_{1,2}&=(\sqrt3/2,0,\mp1/2),&
 \alpha_3&=(0,-1,0).
\end{aligned}
\]
The correlation matrix of $|\Phi^+\rangle$ is
$\operatorname{diag}(1,-1,1)$, and direct substitution gives
$F=5$, hence $I_{3322}=1/4$.
\end{proof}

\section{Verification of the matching identities}
\label{app:matching-identities}

This appendix verifies the two identities used in the proof of
Lemma~\ref{lem:finite-sequence-averaging}.  Assume that every row and
column marginal of $\theta$ is positive, and define the coefficients by
Eq.~\eqref{eq:all-index-coefficients}.

Fix a position $i$ and a pair $(a,b)$, and sum over all sequences
with $(c_{i-1}^\gamma,c_i^\gamma)=(a,b)$.  The column-normalized factors
to the left collapse successively to one, as do the row-normalized factors
to the right.  Only the central factor $\theta_{ab}$ remains, which
proves Eq.~\eqref{eq:target-diagonal-total}.

For $1\leq i<N$, the formulas for the squared coefficients at positions
$i$ and $i+1$ of the same sequence contain the same matrix entries.
They differ only in the normalization at the shared spectral label
$c_i^\gamma$, and therefore
\begin{equation}
 R_{c_i^\gamma}(\lambda_i^\gamma)^2
 =C_{c_i^\gamma}(\lambda_{i+1}^\gamma)^2.
\label{eq:adjacent-ratio}
\end{equation}
For sequences with $c_i^\gamma=c$, taking square roots gives
\begin{equation}
 \lambda_i^\gamma\lambda_{i+1}^\gamma
 =\sqrt{\frac{C_c}{R_c}}\,(\lambda_{i+1}^\gamma)^2.
\end{equation}
Summing this identity and applying
Eq.~\eqref{eq:target-diagonal-total} at position $i+1$ yields
\begin{equation}
 \sum_{\gamma:\,c_i^\gamma=c}
 \lambda_i^\gamma\lambda_{i+1}^\gamma
 =\sqrt{\frac{C_c}{R_c}}\,R_c
 =\sqrt{C_cR_c},
\end{equation}
which is Eq.~\eqref{eq:target-neighbor-total}.  Finally, for a fixed
sequence $\gamma$, every matrix entry
$\theta_{c_{j-1}^\gamma c_j^\gamma}$ occurs in every
$(\lambda_i^\gamma)^2$.  Hence, for each fixed sequence $\gamma$, either
every coefficient $\lambda_i^\gamma$ is positive or every
$\lambda_i^\gamma$ vanishes, justifying the removal of the zero chains in
the main proof.

\section{Technical cases in the strictness argument}
\label{app:technical-cases}

This appendix proves the coarse interval in
Eq.~\eqref{eq:coarse-bounds} and treats the endpoint $c=\pm1$, where
the appending move in Sec.~\ref{sec:strictness} has no linear term.

\paragraph{Proof of Eq.~\eqref{eq:coarse-bounds}.}
For the lower bound, choose a length-29 chain.  Set
\begin{align*}
 c_1=\cdots=c_{12}&=\frac{15}{17},
 &c_{13}&=\frac{171}{221},
 &c_{14}&=\frac5{13},\\
 c_{29-i}&=-c_i &&(1\leq i\leq14),
\end{align*}
with $c_0=1$ and $c_{29}=-1$.  These labels have rational values
of $s$:
\[
 s\left(\frac{15}{17}\right)=\frac8{17},\qquad
 s\left(\frac{171}{221}\right)=\frac{140}{221},\qquad
 s\left(\frac5{13}\right)=\frac{12}{13}.
\]
Choose the unnormalized Schmidt coefficients
\begin{align*}
 (\lambda_1,\ldots,\lambda_{15})
 &=(2,3,\ldots,14,13,12),\\
 \lambda_{30-i}&=\lambda_i &&(1\leq i\leq14).
\end{align*}
Here $\sum_i\lambda_i^2=2510$.  Since all numbers are rational, exact
substitution in Eq.~\eqref{eq:pv-value} immediately gives
\[
 \mathcal P_{29}
 =\frac{15324577}{61295455}
 =\frac14+\frac{2853}{245181820}>\frac14.
\]

For the upper bound, use
\[
 s(c_i)\lambda_i\lambda_{i+1}
 \leq\frac{s(c_i)}2(\lambda_i^2+\lambda_{i+1}^2).
\]
The value in Eq.~\eqref{eq:pv-value} is then at most the largest of
\[
 d(a,b)+\frac{s(a)+s(b)}2,
 \qquad a,b\in[-1,1].
\]
Writing $a=\cos A$, $b=\cos B$, and using product-to-sum identities
gives the elementary bound
\[
 d(a,b)+\frac{s(a)+s(b)}2
 \leq\frac{\sqrt5-1}{4}<\frac13.
\]
For completeness, put $m=(A+B)/2$, $v=(A-B)/2$, and
$x=\sin m$.  The left-hand side becomes
\[
 -x^2-\sin^2v+x(\cos v-\sin v),
\]
whose maximum over $x$ is at most
\[
 -\frac14+\frac12\cos(2v)-\frac14\sin(2v)
 \leq\frac{\sqrt5-1}{4}.
\]
This proves Eq.~\eqref{eq:coarse-bounds}.

\paragraph{The endpoint $c_n=\pm1$.}
At such an endpoint the nearest-neighbor coefficient $s(c_n)$
vanishes, so
the simple appending move has no linear term.  Normalize
the finite or one-sided auxiliary sequence, write $V$ for its value and
$n$ for its last position, and write $c_n=\sigma=\pm1$,
$a=c_{n-1}$.  Set
\[
 c_n(\delta)=\sigma(1-\delta),
 \qquad
 c_{n+1}=\sigma,
\]
and append the auxiliary coefficient
$\alpha\sqrt\delta$, where
$\alpha=u_n/(\sqrt2V)$.  Direct expansion gives
\begin{align*}
 d(a,c_n(\delta))-d(a,\sigma)&=\sigma(1/2-a)\delta,\\
 s(c_n(\delta))&=\sqrt{2\delta-\delta^2},\\
 d(c_n(\delta),\sigma)&=-\left(1+\frac\sigma2\right)\delta,
\end{align*}
and hence the normalized value becomes
\[
 V+\delta u_n^2
 \left[\sigma(1/2-a)+\frac1{2V}\right]+O(\delta^2).
\]
If $\sigma=1$, the first term in brackets is at least $-1/2$; if
$\sigma=-1$, it is at least $-3/2$.  Since $V<1/3$, the bracket
is positive in either case.  The perturbation changes only finitely many
terms and therefore applies unchanged to a one-sided square-summable
auxiliary sequence.

\section{Lean formalization}
\label{app:lean}

The two principal Lean theorem declarations are \nolinkurl{variational} and
\nolinkurl{finiteDimensional_nonattainment}; they prove
$\Istar=\betaPV$ and finite-dimensional nonattainment,
respectively.  The formalization treats finite-dimensional pure-state
strategies with three binary projective measurements per party and Born-rule
values.  The paper's reduction from mixed states and POVMs is not formalized.

The formal upper bound uses a different derivation: it keeps the two
spectral decompositions separate, applies Cauchy--Schwarz on each side,
and then symmetrizes the joint weights.  The source, build instructions,
and dependency graph are available at
\url{https://github.com/jefpauwels/i3322-lean}; the revision audited for this
paper is commit \texttt{f9825cd}.  The source contains no
\texttt{sorry}, \texttt{admit}, or custom axiom; the audit of both
theorems reports only \texttt{propext}, \texttt{Classical.choice}, and
\texttt{Quot.sound}.  The commands \texttt{lake build I3322} and
\texttt{lake env lean Audit.lean} reproduce the build and audit.
 
\end{document}